%% file: wie.tex
\newcommand{\typeof}{0}
\documentclass[10pt,a4paper]{article}

\usepackage{ifthen}
\newcommand{\condinc}[2]{\ifthenelse{\equal{\typeof}{0}}{#1}{#2}}

\input{macros}

\condinc{
\title{On Constructor Rewrite Systems and the Lambda-Calculus}
\author{{Ugo Dal Lago\footnote{
Dipartimento di Scienze dell'Informazione, Universit\`a di Bologna, 
Mura Anteo Zamboni 7, 40127 Bologna, Italy.
\texttt{dallago@cs.unibo.it}
}
}
\and 
{Simone Martini\footnote{
Dipartimento di Scienze dell'Informazione, Universit\`a di Bologna,
Mura Anteo Zamboni 7, 40127 Bologna, Italy.
\texttt{martini@cs.unibo.it}
}
}}}
{
\title{On Constructor Rewrite Systems\\ and the Lambda-Calculus\thanks{The 
authors are partially supported by PRIN project ``CONCERTO'' and FIRB grant RBIN04M8S8, ``Intern. Inst. for Applicable Math.''}}
\author{Ugo Dal Lago\and Simone Martini}
\institute{{Dipartimento di Scienze dell'Informazione, Universit\`a di Bologna\\
            Mura Anteo Zamboni 7, 40127 Bologna, Italy\\
            \email{\{dallago,martini\}@cs.unibo.it
            }}}
}

\begin{document}
\maketitle
\begin{abstract}\noindent
We prove that orthogonal constructor term rewrite systems and lambda-calculus with weak 
(i.e., no reduction is allowed under  the scope of a lambda-abstraction)
call-by-value reduction can simulate each other with a linear overhead.
In particular, weak call-by-value beta-reduction can be simulated by an orthogonal 
constructor term rewrite system in the same number of reduction steps. 
Conversely, each reduction in an term rewrite system can be simulated by a constant number 
of beta-reduction steps. 
This is relevant to implicit computational complexity, because the number of beta steps 
to normal form is polynomially related to the actual cost (that is, as performed on a Turing 
machine) of normalization, under weak call-by-value reduction. Orthogonal constructor term rewrite systems 
and lambda-calculus are thus both polynomially related to Turing machines, taking as notion 
of cost their natural parameters.
\end{abstract}
\section{Motivations}
\par
Implicit computational complexity is a young research area, whose main aim is the description of complexity phenomena 
based on language restrictions, and not on external 
measure conditions or on explicit machine models. 
It borrows techniques and results from mathematical logic (model theory, recursion theory, and proof theory)
and in doing so it has allowed the
incorporation of aspects of computational complexity into areas such as formal methods in software development
and programming language design. The most developed area of implicit computational complexity is probably the
model theoretic one -- finite model theory being a very successful way to describe
complexity classes. In the design of  programming language tools (e.g., type systems), however, syntactical
techniques prove more useful. In the last years we have seen much work restricting
recursive schemata and developing general proof theoretical techniques to enforce
resource bounds on programs. 
Important achievements have been the characterizations of several complexity classes
by means of limitations of recursive definitions (e.g.,~\cite{Bellantoni92CC,Leivant95RRI}) and, more recently, by 
using the ``\emph{light}'' fragments of 
linear logic~\cite{Girard98ic}. 
Moreover, rewriting techniques such as recursive path orderings and
the interpretation method have recently been proved useful in the field~\cite{Marion00}.
By borrowing the terminology from software design technology, we may dub this 
area as implicit computational complexity \emph{in the large}, aiming at a broad, global view on complexity classes.
We may have also an implicit computational complexity \emph{in the small} ---
using logic to study single machine-free models of computation. Indeed, many models of computations do not come
with a natural cost model --- a definition of cost which is both intrinsically rooted in the model of 
computation, and, at the same time, it is polynomially related to the cost of implementing that
model of computation on a standard Turing machine. The main example is the $\lambda$-calculus: The most natural intrinsic
parameter of a computation is its number of beta-reductions, but this very parameter bears no
relation, in general, with the actual cost of performing that computation, since a beta-reduction may involve the duplication of arbitrarily big subterms\footnote{
In full beta-reduction, the size of the duplicated term is indeed arbitrary and does not depend on
the size of the original term the reduction started from. The situation is much different with weak 
reduction, as we will see.}.
What we call implicit computational complexity in the small, therefore, gives complexity significance to
notions and results for computation models where such natural cost measures do not exist, or are
not obvious. In particular, it looks for cost-explicit simulations between such computational models.

The present paper applies this viewpoint to the relation between $\lambda$-calculus and orthogonal (constructor) term rewrite systems.
We will prove that these two machine models simulate each other with a linear overhead. 
That each constructor term rewrite system could be simulated by $\lambda$-terms and beta-reduction is
well known, in view of the availability, in $\lambda$-calculus, of
fixed-point operators, which may be used to solve the mutual recursion expressed by 
first-order rewrite rules. 
Here (Section~\ref{Sect:CTR2L})
we make explicit the complexity content of this simulation, by showing that
any first-order rewriting of $n$ steps can be simulated by $kn$ beta steps, where
$k$ depends on the specific rewrite system but \emph{not} on the size of the involved terms.
Crucial to this result is the encoding of constructor terms using Scott's schema for numerals~\cite{Wadsworth80}.
Indeed, Parigot~\cite{Parigot89CSL} (see also~\cite{ParigotRoziere93}) shows that in the pure $\lambda$-calculus
Church numerals do not admit a predecessor working in a constant number of beta steps.
Moreover, 
Splawski and Urzyczyn~\cite{SplawskiU99} show that it is unlikely that our encoding could work in the typed context of System F.

Section~\ref{sect:CRS} studies the converse -- the simulation of (weak) $\lambda$-calculus reduction by means of
orthogonal constructor term rewrite systems. We give an encoding of $\lambda$-terms into a (first-order) constructor term rewrite 
system. We write   $\LambdatoTRS{\cdot}$ for the map returning a first-order term, given a $\lambda$-term; 
$\LambdatoTRS{M}$ is, in a sense, a complete defunctionalization of the $\lambda$-term $M$, 
where any $\lambda$-abstraction is represented by an atomic constructor. This is similar, 
although not technically the same, to the use of supercombinators (e.g.,~\cite{PJ87}).
We show that $\lambda$-reduction is simulated step by step by first-order rewriting 
(Theorem~\ref{theo:termreducible}).

As a consequence, taking the number of beta steps as a cost model for weak $\lambda$-calculus
is equivalent (up to a linear function) to taking the number of rewritings in orthogonal constructor term rewrite
systems. This is relevant to implicit computational complexity ``in the small'', because
the number of beta steps to normal form is polynomially related to the
actual cost (that is, as performed on a Turing machine) of normalization,
under weak call-by-value reduction. This has been established by 
Sands, Gustavsson, and Moran~\cite{Sands:Lambda02}, by a fine analysis of a $\lambda$-calculus implementation based on
a stack machine. Constructor term rewrite systems and $\lambda$-calculus
are thus both \emph{reasonable} machines (see the ``invariance thesis'' in~\cite{vanEmdeBoas90}),
taking as notion of cost their natural, intrinsic parameters.

As a byproduct, in \condinc{Section~\ref{Sect:GraphRep}}{Section~\ref{Sect:Byproduct}} we sketch a different proof
of the cited result in~\cite{Sands:Lambda02}. Instead of using a stack machine, 
we show how we could encode
constructor term rewriting in term graph rewriting. In term graph rewriting we
avoid the explicit duplication and substitution inherent to
rewriting (and thus also to beta-reduction) and, moreover, we exploit the possible sharing of subterms. 
A more in-depth study of the complexity of 
(constructor) graph rewriting and its relations with (constructor) term rewriting can be found in our~\cite{DLM09}.

\condinc{In Section~\ref{sect:HeadReduction}, we show how to obtain the same results of the previous sections when call-by-name replaces call-by-value as the underlying strategy in the
lambda-calculus.

This paper is an extended version of the one with the same title appeared in the proceedings of ICALP 2009~\cite{DLMicalp}. Besides including full proofs, it has an extended Section~\ref{Sect:GraphRep} and the new material of
Section~\ref{sect:HeadReduction}.}
{}
 
\condinc{}{An extended version of this paper with full proofs is available~\cite{longversion}, where we
describe also an extension of our results 
to $\lambda$-calculus with call-by-name.}
\section{Preliminaries}
The language we study is the pure untyped
$\lambda$-calculus endowed with weak (that is, we never reduce
under an abstraction)
call-by-value reduction. 
\begin{definition}
The following definitions are standard:
\begin{varitemize}
  \item
    \emph{Terms} are defined as follows:
    $$
    \lambdaone::=\varone\;|\;\lambda\varone.\lambdaone\;|\;\lambdaone\lambdaone ,
    $$
    where $\varone$ ranges a denumerable set $\Variables$.
    $\Lambdaterms$ denotes the set of all $\lambda$-terms.
    We assume the existence of a fixed, total, order on $\Variables$; this
    way $\FV{\lambdaone}$ will be a sequence (without repetitions) of variables, not a set. A term
    $\lambdaone$ is said to be \emph{closed} if $\FV{\lambdaone}=\varepsilon$,
    where $\varepsilon$ is the empty sequence.
  \item
    Values are defined as follows:
    $$
    \valueone::=\varone\;|\;\lambda\varone.\lambdaone .
    $$
  \item
    Weak call-by-value reduction is denoted by $\rewrlambdav$
    and is obtained by closing call-by-value reduction under
    any applicative context:
    $$
    \begin{array}{ccccccccc}
      \infer{(\lambda\varone.\lambdaone)\valueone\rewrlambdav\lambdaone\{\valueone/\varone\}}{}
      &&&&
      \infer{\lambdaone\lambdathree\rewrlambdav\lambdatwo\lambdathree}{\lambdaone\rewrlambdav\lambdatwo}
      &&&&
      \infer{\lambdathree\lambdaone\rewrlambdav\lambdathree\lambdatwo}{\lambdaone\rewrlambdav\lambdatwo}
    \end{array}
    $$
    Here $\lambdaone$ ranges over terms, while $\valueone$ ranges over values.
  \item
    The length $\length{M}$ of $M$ is defined as follows, by induction
    on $M$: $\length{\varone}=1$, $\length{\lambda\varone.\lambdaone}=\length{\lambdaone}+1$ 
    and $\length{\lambdaone\lambdatwo}=\length{\lambdaone}+\length{\lambdatwo}+1$.
\end{varitemize}
\end{definition}
Weak call-by-value reduction enjoys many nice properties. In particular, the
one-step diamond property holds and, as a consequence, the number of beta steps
to normal form (if any) is invariant on the reduction order~\cite{CIE2006} (this
justifies the way we defined reduction, which is slightly more general than
Plotkin's one~\cite{Plotkin75tcs}). It is
then meaningful to define $\Time{\lambdaone}$ as \emph{the} number of beta steps to
normal form (or $\omega$ if such a normal form does not exist). This cost model will
be referred to as the \emph{unitary} cost model, since each beta (weak call-by-value) reduction step
counts for $1$ in the global cost of normalization. Moreover, notice that 
$\alpha$-conversion is not needed during reduction of closed terms: if $\lambdaone\rewrlambdav\lambdatwo$
and $\lambdaone$ is closed, then the reduced redex will be in the form $(\lambda\varone.\lambdathree)\valueone$,
where $\valueone$ is a \emph{closed} value. As a consequence, arguments are always closed and open variables
cannot be captured.

The following lemma gives us a generalization of the fixed-point (call-by-value)
combinator (but observe the explicit limit $k$ on the reduction length, 
in the spirit of implicit computational complexity in the small):
\begin{lemma}\label{lemma:mfpc}
For every natural number $n$, there are terms $H_1,\ldots,H_n$ and a natural number $m$ such
that for any sequence of values $V_1,\ldots,V_n$ and for any $1\leq i\leq n$:
$$
H_iV_1\ldots V_n\rewrlambdav^k V_i(\lambda x.H_1V_1\ldots V_nx)\ldots(\lambda x.H_nV_1\ldots V_nx),
$$
where $k\leq m$.
\end{lemma}
\condinc{
\begin{proof}
The terms we are looking for are simply the following:
$$
H_i\equiv\lambdaone_i\lambdaone_1\ldots\lambdaone_n
$$
where, for every $1\leq j\leq n$,
$$
\lambdaone_j\equiv\lambda\varone_1.\ldots.\lambda\varone_n.\lambda\vartwo_1.\ldots.\vartwo_n.\vartwo_j
(\lambda\varthree.\varone_1\varone_1\ldots\varone_n\vartwo_1\ldots\vartwo_n\varthree)
\ldots
(\lambda\varthree.\varone_n\varone_1\ldots\varone_n\vartwo_1\ldots\vartwo_n\varthree).
$$
The natural number $m$ is simply $2n$.
\end{proof}
}
{
The proof of Lemma~\ref{lemma:mfpc} can be found in \cite{longversion}.
}

We will consider in this paper orthogonal constructor (term)
rewrite systems (CRS, see~\cite{TGRterese}). A constructor (term) rewrite system
is a pair $\TRSone=(\Functions{\TRSone},\Rules{\TRSone})$ where:
\begin{varitemize}
\item
  Symbols in the signature $\Functions{\TRSone}$ can be either
  \emph{constructors} or \emph{function symbols}, each with its arity.
  \begin{varitemize} 
    \item
      Terms in $\TRScontermsp{\TRSone}$ are those built
      from constructors and are called \emph{constructor terms}.
    \item
      Terms in $\TRSpatsp{\TRSone}$ are those built
      from constructors and variables and are called \emph{patterns}.
    \item
      Terms in $\TRStermsp{\TRSone}$ are those built
      from constructor and function symbols and are called \emph{closed terms}.
    \item
      Terms in $\TRSvartermsp{\TRSone}$ are those built
      from constructors, functions symbols and variables in $\Variables$ and are dubbed
      \emph{terms}.
    \end{varitemize}
\item
  Rules in $\Rules{\TRSone}$ are in the form $\funone(\patone_1,\ldots,\patone_n)\rewr{\TRSone}\termone$
  where $\funone$ is a function symbol, $\patone_1,\ldots,\patone_n\in\TRSpatsp{\TRSone}$
  and $t\in\TRSvartermsp{\TRSone}$.
  We here consider orthogonal rewrite systems only, i.e. we assume that no distinct two
  rules in $\Rules{\TRSone}$ are overlapping and that every variable appears at most
  once in the lhs of any rule in $\Rules{\TRSone}$. Moreover, we assume that reduction is
  call-by-value, i.e. the substitution triggering any reduction must assign
  constructor terms to variables. This restriction is anyway natural in
  constructor rewriting.
\end{varitemize}
For any term $\termone$ in a CRS, $\length{\termone}$ denotes 
the number of symbol occurrences\condinc{, while $\plength{\termone}{\funone}$ denotes
the number of occurrences of the symbol $\funone$ in $\termone$.}{.}
\section{From Lambda-Calculus to Constructor Term Rewriting}\label{sect:CRS}
\begin{definition}[The CRS $\TRS$]
The constructor rewrite system $\TRS$ is defined as a set of rules $\Rules{\TRS}$ over
an infinite signature $\Functions{\TRS}$. In particular:
\begin{varitemize} 
  \item
    The signature $\Functions{\TRS}$ includes the binary function symbol $\appTRS$ and
    constructor symbols $\constr{\varone}{M}$ for every $\lambdaone\in\Lambdaterms$ and
    every $\varone\in\Variables$. The arity of
    $\constr{\varone}{\lambdaone}$ is the length of $\FV{\lambda\varone.\lambdaone}$. 
    To every term $\lambdaone\in\Lambdaterms$ we can associate
    a term $\LambdatoTRS{\lambdaone}\in\TRSvarterms$ as follows:
    \begin{eqnarray*}
      \LambdatoTRS{\varone}&=&\varone;\\
      \LambdatoTRS{\lambda\varone.\lambdaone}&=&\constr{\varone}{\lambdaone}(\varone_1,\ldots,\varone_{n}),
        \mbox{ where $\FV{\lambda\varone.\lambdaone}=\varone_1,\ldots,\varone_n$};\\
      \LambdatoTRS{\lambdaone\lambdatwo}&=&\appTRS(\LambdatoTRS{\lambdaone},\LambdatoTRS{\lambdatwo}).
    \end{eqnarray*}
    Observe that if $\lambdaone$ is closed, then $\LambdatoTRS{\lambdaone}\in\TRSterms$. 
  \item
    The rewrite rules in $\Rules{\TRS}$ are all the rules in the following form:
    $$
    \appTRS(\constr{\varone}{\lambdaone}(\varone_1,\ldots,\varone_{n}),x)\rewrTRS\LambdatoTRS{M},
    $$
    where $\FV{\lambda\varone.\lambdaone}=\varone_1,\ldots,\varone_n$.
  \item
    A term $\termone\in\TRSterms$ is \emph{canonical} if either $\termone\in\TRSconterms$ or
    $\termone=\appTRS(\termtwo,\termthree)$ where $\termtwo$ and $\termthree$ are
    themselves canonical.
  \end{varitemize}
\end{definition}
Notice that the signature $\Functions{\TRS}$ contains an infinite amount of constructors.

\begin{example}
\condinc{}{\begin{sloppypar}}
Consider the $\lambda$-term $\lambdaone=(\lambda \varone.\varone\varone)(\lambda \vartwo.\vartwo\vartwo)$.
$\LambdatoTRS{\lambdaone}$ is
$\termone\equiv\appTRS(\constr{\varone}{\varone\varone},\constr{\vartwo}{\vartwo\vartwo})$.
Moreover, $\termone\rewrTRS\appTRS(\constr{\vartwo}{\vartwo\vartwo},\constr{\vartwo}{\vartwo\vartwo})\equiv\termtwo$, 
as expected. Finally, we have $\termtwo\rewrTRS\termtwo$.
\condinc{}{\end{sloppypar}}
\end{example}
To any term in $\TRSvarterms$ corresponds a $\lambda$-term in $\Lambdaterms$:
\begin{definition}
  To every term $\termone\in\TRSvarterms$ we can associate
  a term $\TRStolambda{\termone}\in\Lambdaterms$ as follows:
  \condinc{
  \begin{eqnarray*}
    \TRStolambda{\varone}&=&\varone\\
    \TRStolambda{\appTRS(\termtwo,\termthree)}&=&\TRStolambda{\termtwo}\TRStolambda{\termthree}\\
    \TRStolambda{\constr{\varone}{\lambdaone}(\termone_1,\ldots\termone_n)}&=&
    (\lambda\varone.\lambdaone)\{\TRStolambda{\termone_1}/\varone_1,\ldots,\TRStolambda{\termone_n}/\varone_n\}
  \end{eqnarray*}}{
    $\TRStolambda{\varone}=\varone$, 
    $\TRStolambda{\appTRS(\termtwo,\termthree)}=\TRStolambda{\termtwo}\TRStolambda{\termthree}$ and
    $\TRStolambda{\constr{\varone}{\lambdaone}(\termone_1,\ldots\termone_n)}=
    (\lambda\varone.\lambdaone)\{\TRStolambda{\termone_1}/\varone_1,\ldots,\TRStolambda{\termone_n}/\varone_n\}$}
  where $\FV{\lambda\varone.\lambdaone}=\varone_1,\ldots,\varone_n$.
\end{definition}
Canonicity holds for terms in $\TRS$ obtained as images of (closed) $\lambda$-terms via $\LambdatoTRS{\cdot}$. Moreover, canonicity
is preserved by reduction in $\TRS$:
\begin{lemma}\label{lemma:canonicity}
For every closed $\lambdaone\in\Lambdaterms$, $\LambdatoTRS{\lambdaone}$ is canonical.
Moreover, if $\termone$ is canonical and $\termone\rewrTRS\termtwo$, 
then $\termtwo$ is canonical.
\end{lemma}
\condinc{
\begin{proof}
$\LambdatoTRS{\lambdaone}$ is canonical for any $\lambdaone\in\Lambdaterms$ by
induction on the structure of $\lambdaone$ (which, by hypothesis, is either an
abstraction or an application $\lambdatwo\lambdathree$ where both
$\lambdatwo$ and $\lambdathree$ are closed). We can further prove that
$\termthree=\LambdatoTRS{\lambdaone}\{\termone_1/\varone_1,\ldots\termone_n/\varone_n\}$
is canonical whenever $\termone_1,\ldots,\termone_n\in\TRSconterms$ and
$\varone_1,\ldots,\varone_n$ includes all the variables in $\FV{\lambdaone}$:
\begin{varitemize}
\item
  If $\lambdaone=\varone_i$, then $\termthree=\termone_i$, which is
  clearly canonical.
\item
  If $\lambdaone=\lambdatwo\lambdathree$, then
  \begin{eqnarray*}
  \termthree&=&\LambdatoTRS{\lambdatwo\lambdathree}\{\termone_1/\varone_1,\ldots\termone_n/\varone_n\}\\
     &=&\appTRS\left(\LambdatoTRS{\lambdatwo}\{\termone_1/\varone_1,\ldots\termone_n/\varone_n\},
        \LambdatoTRS{\lambdathree}\{\termone_1/\varone_1,\ldots\termone_n/\varone_n\}\right)
  \end{eqnarray*}
  which is canonical, by IH.
\item
  If $\lambdaone=\lambda\vartwo.\lambdatwo$, then
  \begin{eqnarray*}
  \termthree&=&\LambdatoTRS{\lambda\vartwo.\lambdatwo}\{\termone_1/\varone_1,\ldots\termone_n/\varone_n\}\\
     &=&\constr{\vartwo}{\lambdatwo}(\varone_{i_1},\ldots,\varone_{i_m})\{\termone_1/\varone_1,\ldots\termone_n/\varone_n\}\\
     &=&\constr{\vartwo}{\lambdatwo}(\termone_{i_1},\ldots,\termone_{i_m})
  \end{eqnarray*}
  which is canonical, because each $\termone_i$ is in $\TRSconterms$.
\end{varitemize}
This implies the rhs of any instance of a rule in $\Rules{\TRS}$ is canonical. As a consequence,
$\termtwo$ is canonical whenever $\termone\rewrTRS\termtwo$ and $\termone$ is canonical.
This concludes the proof. 
\end{proof}}{}
For canonical terms, being a normal form is equivalent of being mapped to a normal
form via $\TRStolambda{\cdot}$. This is not true, in general: take as a counterexample
$\constr{\varone}{\vartwo}(\appTRS(\constr{\varthree}{\varthree},\constr{\varthree}{\varthree}))$,
which corresponds to $\lambda\varone.(\lambda\varthree.\varthree)(\lambda\varthree.\varthree)$ via $\TRStolambda{\cdot}$.
\begin{lemma}\label{lemma:NFcanonical}
A canonical term $\termone$ is a normal form iff 
$\TRStolambda{\termone}$ is a normal form.
\end{lemma}
\condinc{
\begin{proof}
If a canonical $\termone$ is a normal form, then $\termone$ does not contain the function symbol $\appTRS$ and,
as a consequence, $\TRStolambda{\termone}$ is an abstraction, which is always a normal form.
Conversely, if $\TRStolambda{\termone}$ is a normal form, then
$\termone$ is not in the form $\appTRS(\termtwo,\termthree)$, because otherwise
$\TRStolambda{\termone}$ will be a (closed) application, which cannot be a normal form.
But since $\termone$ is canonical, $\termone\in\TRSconterms$, which only contains terms
in normal form.
\end{proof}}{}
\condinc{
The following substitution lemma will be useful later.
\begin{lemma}
For every term $\termone\in\TRSvarterms$ and every $\termone_1,\ldots,\termone_n\in\TRSconterms$,
$$
\TRStolambda{\termone\{\termone_1/\varone_1,\ldots,\termone_n/\varone_n\}}=
\TRStolambda{\termone}\{\TRStolambda{\termone_1}/\varone_1,\ldots,\TRStolambda{\termone_n}/\varone_n\}
$$
whenever $\varone_1,\ldots,\varone_n$ includes all the variables in $\termone$.
\end{lemma}
\begin{proof}
By induction on $\termone$:
\begin{varitemize}
  \item
    If $\termone=\varone_i$, then 
    \begin{eqnarray*}
      \TRStolambda{\termone\{\termone_1/\varone_1,\ldots,\termone_n/\varone_n\}}&=&
      \TRStolambda{\varone_i\{\termone_1/\varone_1,\ldots,\termone_n/\varone_n\}}\\
      &=&\TRStolambda{\termone_i}\\
      &=&\varone_i\{\TRStolambda{\termone_1}/\varone_1,\ldots,\TRStolambda{\termone_n}/\varone_n\}\\
      &=&\termone\{\TRStolambda{\termone_1}/\varone_1,\ldots,\TRStolambda{\termone_n}/\varone_n\}.
    \end{eqnarray*}
  \item
    If $\termone=\appTRS(\termtwo,\termthree)$, then
    \begin{eqnarray*}
      \TRStolambda{\termone\{\termone_1/\varone_1,\ldots,\termone_n/\varone_n\}}&=&
      \TRStolambda{\appTRS(\termtwo,\termthree)\{\termone_1/\varone_1,\ldots,\termone_n/\varone_n\}}\\
      &=&\TRStolambda{\appTRS(\termtwo\{\termone_1/\varone_1,\ldots,\termone_n/\varone_n\},
        \termthree\{\termone_1/\varone_1,\ldots,\termone_n/\varone_n\})}\\
      &=&\TRStolambda{\termtwo\{\termone_1/\varone_1,\ldots,\termone_n/\varone_n\}}
      \TRStolambda{\termthree\{\termone_1/\varone_1,\ldots,\termone_n/\varone_n\}}\\
      &=&\TRStolambda{\termtwo}\{\TRStolambda{\termone_1}/\varone_1,\ldots,\TRStolambda{\termone_n}/\varone_n\}
      \TRStolambda{\termthree}\{\TRStolambda{\termone_1}/\varone_1,\ldots,\TRStolambda{\termone_n}/\varone_n\}\\
      &=&\TRStolambda{\termtwo}\TRStolambda{\termthree}
      \{\TRStolambda{\termone_1}/\varone_1,\ldots,\TRStolambda{\termone_n}/\varone_n\}\\
      &=&\TRStolambda{\appTRS(\termtwo,\termthree)}\{\TRStolambda{\termone_1}/\varone_1,\ldots,\TRStolambda{\termone_n}/\varone_n\}\\
      &=&\TRStolambda{\termone}\{\TRStolambda{\termone_1}/\varone_1,\ldots,\TRStolambda{\termone_n}/\varone_n\}.
    \end{eqnarray*}
  \item
    If $\termone=\constr{\vartwo}{\lambdatwo}(\termtwo_{1},\ldots,\termtwo_{m})$, then
    \begin{eqnarray*}
      \TRStolambda{\termone\{\termone_1/\varone_1,\ldots,\termone_n/\varone_n\}}&=&
      \TRStolambda{\constr{\vartwo}{\lambdatwo}(\termtwo_{1},\ldots,\termtwo_{m})
        \{\termone_1/\varone_1,\ldots,\termone_n/\varone_n\}}\\
      &=&\TRStolambda{\constr{\vartwo}{\lambdatwo}
        (\termtwo_{1}\{\termone_1/\varone_1,\ldots,\termone_n/\varone_n\},\ldots,
         \termtwo_{m}\{\termone_1/\varone_1,\ldots,\termone_n/\varone_n\})}\\
      &=&(\lambda\vartwo.\lambdatwo)\{
         \TRStolambda{\termtwo_{1}\{\termone_1/\varone_1,\ldots,\termone_n/\varone_n\}}/\varone_{i_1}\\
         & &\hspace{36pt},\ldots,\\
         & &\hspace{36pt}\TRStolambda{\termtwo_{m}\{\termone_1/\varone_1,\ldots,\termone_n/\varone_n\}}/\varone_{i_m}\}\\
      &=&(\lambda\vartwo.\lambdatwo)\{\TRStolambda{\termtwo_{1}}\{\TRStolambda{\termone_1}/\varone_1,\ldots,
         \TRStolambda{\termone_n}/\varone_n\}/\varone_{i_1}\\
         & &\hspace{36pt},\ldots,\\
         & &\hspace{36pt}\TRStolambda{\termtwo_m}\{\TRStolambda{\termone_1}/\varone_1,\ldots,\TRStolambda{\termone_n}/\varone_n\}/\varone_{i_m}\}\\
      &=&((\lambda\vartwo.\lambdatwo)\{\TRStolambda{\termtwo_{1}}/\varone_1,\ldots,\termtwo_{_m}/\varone_{i_1}\})
         \{\TRStolambda{\termone_1}/\varone_1,\ldots,\TRStolambda{\termone_n}/\varone_n\}\\
      &=&\TRStolambda{\constr{\vartwo}{\lambdatwo}(\termtwo_{1},\ldots,\termtwo_{m})}\{\TRStolambda{\termone_1}/\varone_1,\ldots,\TRStolambda{\termone_n}/\varone_n\}\\
      &=&\TRStolambda{\termone}\{\TRStolambda{\termone_1}/\varone_1,\ldots,\TRStolambda{\termone_n}/\varone_n\}.
    \end{eqnarray*}
\end{varitemize}
This concludes the proof.
\end{proof}
\begin{lemma}\label{lemma:invert}
For every $\lambda$-term $\lambdaone\in\Lambdaterms$, $\TRStolambda{\LambdatoTRS{\lambdaone}}=\lambdaone$.
\end{lemma}
\begin{proof}
By induction on $\lambdaone$:
\begin{varitemize}
\item
  If $\lambdaone=\varone$, then 
  $$
    \TRStolambda{\LambdatoTRS{\lambdaone}}=\TRStolambda{\LambdatoTRS{\varone}}=\TRStolambda{\varone}=\varone.
  $$
\item
  If $\lambdaone=\lambdatwo\lambdathree$, then
  $$
  \TRStolambda{\LambdatoTRS{\lambdaone}}=
  \TRStolambda{\appTRS(\LambdatoTRS{\lambdatwo},\LambdatoTRS{\lambdathree})}=
  \TRStolambda{\LambdatoTRS{\lambdatwo}}\TRStolambda{\LambdatoTRS{\lambdathree}}=
  \lambdatwo\lambdathree.
  $$
\item
  If $\lambdaone=\lambda\vartwo.\lambdatwo$, then
  $$
    \TRStolambda{\LambdatoTRS{\lambdaone}}=
    \TRStolambda{\constr{\vartwo}{\lambdatwo}(\varone_{1},\ldots,\varone_{n})}=
    (\lambda\vartwo.\lambdatwo)\{\varone_1/\varone_1,\ldots,\varone_n/\varone_n\}=
    \lambda\vartwo.\lambdatwo=\lambdaone.
  $$
\end{varitemize}
This concludes the proof.
\end{proof}
The previous two lemmas implies that if $\lambdaone\in\Lambdaterms$, 
$\termone_1,\ldots,\termone_n\in\TRSconterms$ and $\varone_1,\ldots,\varone_n$ 
includes all the variables in $\FV{\lambdaone}$, then:
\begin{equation}\label{equat:commute}
\TRStolambda{\LambdatoTRS{\lambdaone}\{\termone_1/\varone_1,\ldots,\termone_n/\varone_n\}}=
\lambdaone\{\TRStolambda{\termone_1}/\varone_1,\ldots,\TRStolambda{\termone_n}/\varone_n\}.
\end{equation}}{}
Reduction in $\TRS$ can be simulated by reduction in the $\lambda$-calculus,
provided the starting term is canonical.  
\begin{lemma}\label{lemma:TRStolam}
If $\termone$ is canonical and $\termone\rewrTRS\termtwo$, then
$\TRStolambda{\termone}\rewrlambdav\TRStolambda{\termtwo}$.
\end{lemma}
\condinc{
\begin{proof}
Consider the (instance of the) rewriting rule which
turns $\termone$ into $\termtwo$. Let it be
$$
\appTRS(\constr{\vartwo}{\lambdaone}(\termone_1,\ldots,\termone_n),\termthree)\rewrTRS
\LambdatoTRS{\lambdaone}\{\termone_1/\varone_1,\ldots,\termone_n/\varone_n,\termthree/\vartwo\}.
$$
Clearly,
$$
\TRStolambda{\appTRS(\constr{\vartwo}{\lambdaone}(\termone_1,\ldots,\termone_n),\termthree)}=
  ((\lambda\vartwo.\lambdaone)\{\termone_1/\varone_1,\ldots,\termone_n/\varone_n\})\TRStolambda{\termthree}
$$
while, by~(\ref{equat:commute}):
$$
\TRStolambda{\LambdatoTRS{\lambdaone}\{\termone_1/\varone_1,\ldots,\termone_n/\varone_n,\termthree/\vartwo\}}
=\lambdaone\{\TRStolambda{\termone_1}/\varone_1,\ldots,\TRStolambda{\termone_n}/\varone_n,\TRStolambda{\termthree}/\vartwo\}
$$
which implies the thesis.
\end{proof}}{}
Conversely, call-by-value reduction in the $\lambda$-calculus can be simulated in $\TRS$:
\begin{lemma}\label{lemma:lamtoTRS}
If $\lambdaone\rewrlambdav\lambdatwo$, $\termone$ is canonical and $\TRStolambda{\termone}=\lambdaone$, then
$\termone\rewrTRS\termtwo$, where $\TRStolambda{\termtwo}=\lambdatwo$.
\end{lemma}
\condinc{
\begin{proof}
Let $(\lambda\varone.\lambdathree)\valueone$ be the redex fired in $\lambdaone$ when rewriting
it to $\lambdatwo$. There must be a corresponding subterm $\termthree$ of $\termone$ such
that $\TRStolambda{\termthree}=(\lambda\varone.\lambdathree)\valueone$. Then
$$
\termthree=\appTRS(\constr{\varone}{\lambdafour}(\termone_1,\ldots,\termone_n),\termfour),
$$
where $\TRStolambda{\constr{\varone}{\lambdafour}(\termone_1,\ldots,\termone_n)}=\lambda\varone.\lambdathree$.
and $\TRStolambda{\termfour}=\valueone$. Observe that, by definition,
$$
\TRStolambda{\constr{\varone}{\lambdafour}(\termone_1,\ldots,\termone_n)}=
(\lambda\varone.\lambdafour)\{\TRStolambda{\termone_1}/\varone_1,\ldots,\TRStolambda{\termone_n}/\varone_n\}
$$
where $\FV{\lambdafour}=\varone_1,\ldots,\varone_n$. Since $\termone$ is canonical, 
$\termone_1,\ldots,\termone_n\in\TRSconterms$. Moreover, since $\valueone$ is a value,
$\termfour$ itself is in $\TRSconterms$.This implies
$$
\appTRS(\constr{\varone}{\lambdafour}(\termone_1,\ldots,\termone_n),\termfour)\rewrTRS
\LambdatoTRS{\lambdafour}\{\termone_1/\varone_1,\ldots,\termone_n/\varone_n,\termfour/\varone\}.
$$   
By~(\ref{equat:commute}):
\begin{eqnarray*}
\TRStolambda{\LambdatoTRS{\lambdafour}\{\termone_1/\varone_1,\ldots,\termone_n/\varone_n,\termfour/\varone\}}&=&
   \lambdafour\{\TRStolambda{\termone_1}/\varone_1,\ldots,\TRStolambda{\termone_n}/\varone_n,\TRStolambda{\termfour}/\varone\}\\
&=&(\lambdafour\{\TRStolambda{\termone_1}/\varone_1,\ldots,\TRStolambda{\termone_n}/\varone_n\})\{\TRStolambda{\termfour}/\varone\}\\
&=&(\lambda\varone.\lambdathree)\{\valueone/\varone\}.
\end{eqnarray*}
This concludes the proof.
\end{proof}}{}
The previous lemmas altogether imply the following theorem, by which $\lambda$-calculus
normalization can be mimicked (step-by-step) by reduction in $\TRS$:
\begin{theorem}[Term Reducibility]\label{theo:termreducible}
Let $\lambdaone\in\Lambdaterms$ be a closed term. The following
two conditions are equivalent:
\begin{varenumerate}
\item
  $\lambdaone\rewrlambdav^n\lambdatwo$ where $\lambdatwo$ is in normal form;
\item
  $\LambdatoTRS{\lambdaone}\rewrTRS^n\termone$ where
  $\TRStolambda{\termone}=\lambdatwo$ and $\termone$ is in normal form.
\end{varenumerate}
\end{theorem}
\begin{proof}
Suppose $\lambdaone\rewrlambdav^n\lambdatwo$, where $\lambdatwo$ is in normal form. 
Then, by applying Lemma~\ref{lemma:lamtoTRS}, we obtain a term $\termone$ such that
$\LambdatoTRS{\lambdaone}\rewrTRS^n\termone$ and $\TRStolambda{\termone}=\lambdatwo$.
By Lemma~\ref{lemma:canonicity}, $\termone$ is canonical and, by Lemma~\ref{lemma:NFcanonical}, 
it is in normal form. Now, suppose $\LambdatoTRS{\lambdaone}\rewrTRS^n\termone$ where
$\TRStolambda{\termone}=\lambdatwo$ and $\termone$ is in normal form. By
applying $n$ times Lemma~\ref{lemma:TRStolam}, we obtain
$\TRStolambda{\LambdatoTRS{\lambdaone}}\rewrlambdav^n\TRStolambda{\termone}=\lambdatwo$.
But $\TRStolambda{\LambdatoTRS{\lambdaone}}=\lambdaone$ \condinc{by Lemma~\ref{lemma:invert}}{by an
easy induction on $\lambdaone$} and $\lambdatwo$ is a normal form 
by Lemma~\ref{lemma:NFcanonical}, since $\LambdatoTRS{\lambdaone}$ and $\termone$ are
canonical by Lemma~\ref{lemma:canonicity}.\condinc{}{\hfill$\Box$}
\end{proof}
There is another nice property of $\TRS$, that will be crucial in proving the main
result of this paper:
\begin{proposition}\label{prop:constred}
For every $\lambdaone\in\Lambdaterms$, for every $\termone$ with $\LambdatoTRS{\lambdaone}\rewrTRS^*\termone$
and for every occurrence of a constructor $\constr{\varone}{\lambdatwo}$ in $\termone$, $\lambdatwo$ is
a subterm of $\lambdaone$.
\end{proposition}
\condinc{
\begin{proof}
Assume $\LambdatoTRS{\lambdaone}\rewrTRS^n\termone$ and proceed
by induction on $n$.
\end{proof}}{}
\begin{example}
Let us consider the $\lambda$-term $\lambdaone=(\lambda\varone.(\lambda \vartwo.\varone)\varone)
(\lambda\varthree.\varthree)$. Notice that 
$$
\lambdaone\rewrlambdav(\lambda \vartwo.(\lambda\varthree.\varthree))
(\lambda\varthree.\varthree)\rewrlambdav\lambda\varthree.\varthree.
$$
Clearly $\LambdatoTRS{\lambdaone}=\appTRS(\constr{\varone}{(\lambda\vartwo.\varone)\varone},
\constr{\varthree}{\varthree})$. Moreover:
$$
\appTRS(\constr{\varone}{(\lambda\vartwo.\varone)\varone},
\constr{\varthree}{\varthree})\rewrTRS \appTRS(\constr{\vartwo}{\varone}(\constr{\varthree}{\varthree}),
   \constr{\varthree}{\varthree})\rewrTRS\constr{\varthree}{\varthree}.
$$
For every constructor $\constr{\varfour}{\lambdatwo}$ occurring in any term in the previous
reduction sequence, $\lambdatwo$ is a subterm of $\lambdaone$. 
\end{example}

A remark on $\TRS$ is now in order. $\TRS$ is an infinite CRS, since $\Functions{\TRS}$ contains
an infinite amount of constructor symbols and, moreover, there are infinitely many
rules in $\Rules{\TRS}$. As a consequence, what we have presented here is an embedding
of the (weak, call-by-value) $\lambda$-calculus into an infinite (orthogonal) CRS. Consider, now,
the following scenario: suppose the $\lambda$-calculus is used to write a \emph{program} $\lambdaone$,
and suppose that inputs to $\lambdaone$ form an infinite set of $\lambda$-terms $\Theta$ which can 
anyway be represented by a finite set of constructors in $\TRS$. In this scenario,
Proposition~\ref{prop:constred} allows to conclude the existence of finite
subsets of $\Functions{\TRS}$ and $\Rules{\TRS}$ such that \emph{every} $\lambdaone\lambdatwo$
(where $\lambdatwo\in\Theta$) can be reduced via $\TRS$ by using only
constructors and rules in those \emph{finite} subsets. As a consequence, we can see the
above schema as one that puts any program $\lambdaone$ 
in correspondence to a \emph{finite} CRS. Finally, observe that assuming \emph{data} to
be representable by a finite number of constructors in $\TRS$ is reasonable.
Scott's scheme~\cite{Wadsworth80}, for example, allows to represent any term in a given
free algebra in a finitary way, e.g. the natural number $0$ becomes
$\nat{0}\equiv\constr{\vartwo}{\lambda\varthree.\varthree}$ while
$n+1$ becomes $\nat{n+1}\equiv\constr{\vartwo}{\lambda\varthree.\vartwo\varone}(\nat{n})$.
Church's scheme, on the other hand, does not have this property.
\section{From Constructor Term Rewriting to Lambda-Calculus}\label{Sect:CTR2L}
In this Section, we will show that any rewriting step of a constructor rewrite
system can be simulated by a fixed number of weak call-by-value beta-reductions. 

Let $\TRSone$ be an orthogonal constructor rewrite system over a finite signature $\Functions{\TRSone}$. Let
$\conone_1,\ldots,\conone_g$ be the constructors of $\TRSone$
and let $\funone_1,\ldots,\funone_h$ be the function symbols of
$\TRSone$. The following constructions work independently of $\TRSone$.

We will first concentrate on constructor terms, encoding them as
$\lambda$-terms using Scott's schema~\cite{Wadsworth80}.
Constructor terms can be easily put in correspondence with $\lambda$-terms
by way of a map $\TRSonetolambda{\cdot}$ defined by induction as follows:
$$
\TRSonetolambda{\conone_i(\termone_1\ldots,\termone_n)}\equiv
\lambda x_1.\ldots.\lambda x_g.\lambda y.x_i\TRSonetolambda{\termone_1}\ldots\TRSonetolambda{\termone_n}.
$$
This way constructors become functions:
$$
\TRSonetolambda{\conone_i}\equiv\lambda x_1.\ldots.\lambda x_{\arity{\conone_i}}.
 \lambda y_1.\ldots.\lambda y_g.\lambda z.y_ix_1\ldots x_{\arity{\conone_i}}.
$$
Trivially, $\TRSonetolambda{\conone_i}\TRSonetolambda{\termone_1}\ldots\TRSonetolambda{\termone_n}$
rewrites to $\TRSonetolambda{\conone_i(\termone_1\ldots\termone_n)}$ in
$\arity{\conone_i}$ steps. To represent an error value, we use the $\lambda$-term
$\errorterm \equiv
\lambda x_1.\ldots.\lambda x_g.\lambda y.y$. A $\lambda$-term
built in this way, i.e. a $\lambda$-term which is either $\errorterm$ or in
the form $\TRSonetolambda{\termone}$ is denoted with metavariables like
$\cltermone$ or $\cltermtwo$.

The map $\TRSonetolambda{\cdot}$ defines encodings of constructor terms.
But what about terms containing function symbols?
The goal is defining another map $\TRSonetolambdaII{\cdot}$ returning
a $\lambda$-term given any term $t$ in $\TRStermsp{\TRSone}$, in such
a way that $\termone\rewrTRS^*\termtwo$ and $\termtwo\in\TRScontermsp{\TRSone}$ implies
$\TRSonetolambdaII{\termone}\rewrlambdav^*\TRSonetolambda{\termtwo}$. Moreover,
$\TRSonetolambdaII{\termone}$ should rewrite to $\errorterm$ whenever the rewriting of $\termone$
causes an error (i.e. whenever $\termone$ has a normal form containing
a function symbol).
\condinc{
First of all, we can define the $\lambda$-term $\TRSonetolambdaII{\conone_i}$ corresponding
to any constructor $\conone_i$. To do that, define a $\lambda$-term $\lambdaone^i_{\varone_1,\ldots,\varone_m}$
for every $1\leq i\leq g$, for every $0\leq m\leq\arity{\conone_i}$ and for every
variables $\varone_1,\ldots,\varone_m$ by induction on $\arity{\conone_i}-m$:
\begin{eqnarray*}
  \lambdaone^i_{\varone_1,\ldots,\varone_{\arity{\conone_i}}}&\equiv&
  \lambda\vartwo_1.\ldots.\lambda\vartwo_{g}.\vartwo_{i}\varone_1\ldots\varone_{\arity{\conone_i}};\\
  \forall m: 0\leq m<\arity{\conone_i} \qquad \lambdaone^i_{\varone_1,\ldots,\varone_m}&\equiv&
  \lambda\vartwo.\vartwo\lambdatwo_{1,i}^m\ldots\lambdatwo_{g,i}^m\lambdathree^m_i;
\end{eqnarray*}
where:
\begin{eqnarray*}
\lambdatwo_{j,i}^m&\equiv&\lambda\varthree_1.\ldots.\lambda\varthree_{\arity{\conone_j}}.
  (\lambda\varone_{m+1}.\lambdaone^i_{\varone_1,\ldots,\varone_{m+1}})
  \lambdaone^{\arity{\conone_j}}_{\varthree_1,\ldots,\varthree_{\arity{\conone_j}}};\\
\lambdathree^m_i&\equiv&\lambda\varthree_{m+2}.\ldots.\lambda\varthree_{\arity{\conone_i}}.\errorterm.
\end{eqnarray*}
\begin{lemma}
There is a constant $n\in\N$ such that for every $i$ and for every $m$:
$$
\lambdaone^i_{\varone_1,\ldots,\varone_m}\{\TRSonetolambda{\termone_1}/\varone_1,\ldots,\TRSonetolambda{\termone_m}/\varone_m\}
\TRSonetolambda{\termone_{m+1}}\ldots\TRSonetolambda{\termone_{\arity{\conone_i}}}
\rewrTRS^k\TRSonetolambda{\conone_i(\termone_1\ldots\termone_{\arity{\conone_i}})}
$$
(where $k\leq n$) and
$$
\lambdaone^i_{\varone_1,\ldots,\varone_m}\{\TRSonetolambda{\termone_1}/\varone_1,\ldots,\TRSonetolambda{\termone_m}/\varone_m\}
\cltermone_{m+1}\ldots\cltermone_{\arity{\conone_i}}
\rewrTRS^l\errorterm
$$
(where $l\leq n$) whenever $\cltermone_{j}$ is either $\TRSonetolambda{\termone_j}$ or $\errorterm$
but at least one among $\cltermone_{m+1}\ldots\cltermone_{\arity{\conone_i}}$ is $\errorterm$.
\end{lemma}
\begin{proof}
We proceed by induction on $\arity{\conone_i}-m$:
\begin{varitemize}
\item
  If $m=\arity{\conone_i}$, then
  \begin{eqnarray*}
    &&\lambdaone^i_{\varone_1,\ldots,\varone_{\arity{\conone_i}}}\{\TRSonetolambda{\termone_1}/\varone_1,\ldots,\TRSonetolambda{\termone_{\arity{\conone_i}}}/\varone_{\arity{\conone_i}}\}\\
    &\equiv&(\lambda\vartwo_1.\ldots.\lambda\vartwo_{g}\vartwo_{i}\varone_1\ldots\varone_{\arity{\conone_i}})
      \{\TRSonetolambda{\termone_1}/\varone_1,\ldots,\TRSonetolambda{\termone_{\arity{\conone_i}}}/\varone_{\arity{\conone_i}}\}\\
    &\equiv&\lambda\vartwo_1.\ldots.\lambda\vartwo_{g}.\vartwo_{i}\TRSonetolambda{\termone_1}\ldots\TRSonetolambda{\termone_{\arity{\conone_i}}}\\
    &\equiv&\TRSonetolambda{\conone_i(\termone_1,\ldots,\termone_{\arity{\conone_i}})}.
  \end{eqnarray*}
\item
  If $m<\arity{\conone_i}$, we
  use the following abbreviations:
  \begin{eqnarray*}
    \lambdafour_{j,i}^m&\equiv&\lambdatwo_{j,i}^m\{\TRSonetolambda{\termone_1}/\varone_1,\ldots,\TRSonetolambda{\termone_{m}}/\varone_{m}\};\\
    \lambdafive_{j}^m&\equiv&\lambdathree_{j}^m\{\TRSonetolambda{\termone_1}/\varone_1,\ldots,\TRSonetolambda{\termone_{m}}/\varone_{m}\}.
  \end{eqnarray*}
  Let's distinguish two cases:
  \begin{varitemize}
  \item
    If $\cltermone_{m+1}\equiv\errorterm$, then:
    \begin{eqnarray*}
      &&\lambdaone^i_{\varone_1,\ldots,\varone_{m}}\{\TRSonetolambda{\termone_1}/\varone_1,\ldots,
      \TRSonetolambda{\termone_{m}}/\varone_{m}\}\cltermone_{m+1}\ldots\cltermone_{\arity{\conone_i}}\\
      &\rewrlambdav&(\errorterm\lambdafour_{1,i}^m\ldots\lambdafour_{g,i}^m\lambdafive^m_i)\cltermone_{m+2}\ldots\cltermone_{\arity{\conone_i}}\\
      &\rewrlambdav^*&\lambdafive_i^m\cltermone_{m+2}\ldots\cltermone_{\arity{\conone_i}}\\
      &\rewrlambdav^*&\errorterm
    \end{eqnarray*}
  \item
    Let $\cltermone_{m+1}$ be $\TRSonetolambda{\termone_{m+1}}$,
    where $\termone_{m+1}\equiv\conone_j(\termtwo_{1},\ldots,\termtwo_{\arity{\conone_j}}) $. 
    Then:
    \begin{eqnarray*}
      &&\lambdaone^i_{\varone_1,\ldots,\varone_{m}}\{\TRSonetolambda{\termone_1}/\varone_1,\ldots,
        \TRSonetolambda{\termone_{m}}/\varone_{m}\}\cltermone_{m+1}\ldots\cltermone_{\arity{\conone_i}}\\
      &\rewrlambdav&(\TRSonetolambda{\conone_j(\termtwo_{1},\ldots,\termtwo_{\arity{\conone_j}})}
        \lambdafour_{1,i}^m\ldots\lambdafour_{g,i}^m\lambdafive^m_i)\cltermone_{m+2}\ldots\cltermone_{\arity{\conone_i}}\\
      &\rewrlambdav^*&\lambdafour_{j,i}^m\TRSonetolambda{\termtwo_{1}}\ldots\TRSonetolambda{\termtwo_{\arity{\conone_j}}}
        \cltermone_{m+2}\ldots\cltermone_{\arity{\conone_i}}\\
      &\rewrlambdav^*&(\lambda\varone_{m+1}.\lambdaone^i_{\varone_1,\ldots,\varone_{m+1}}\{\TRSonetolambda{\termone_1}/\varone_1,\ldots,
        \TRSonetolambda{\termone_{m}}/\varone_m\})\\
      &&(\lambdaone^{j}_{\varthree_1,\ldots,\varthree_{\arity{\conone_j}}}\{\TRSonetolambda{\termtwo_1}/\vartwo_1,\ldots,
        \TRSonetolambda{\termone_{\arity{\conone_j}}}/\vartwo_{\arity{\conone_j}}\})\cltermone_{m+2}\ldots\cltermone_{\arity{\conone_i}}\\
      &\rewrlambdav^*&(\lambda\varone_{m+1}.\lambdaone^i_{\varone_1,\ldots,\varone_{m+1}}\{\TRSonetolambda{\termone_1}/\varone_1,\ldots,
        \TRSonetolambda{\termone_{m}}/\varone_m\})\\
      &&(\TRSonetolambda{\conone_j(\termtwo_1,\ldots,\termtwo_{\arity{\conone_j}})})
        \cltermone_{m+2}\ldots\cltermone_{\arity{\conone_i}}\\
     &\rewrlambdav^*&\lambdaone^i_{\varone_1,\ldots,\varone_{m+1}}\{\TRSonetolambda{\termone_1}/\varone_1,\ldots,
        \TRSonetolambda{\termone_{m+1}}/\varone_{m+1}\}\cltermone_{m+2}\ldots\cltermone_{\arity{\conone_i}}
    \end{eqnarray*}
    and, by the inductive hypothesis, the last term in the reduction sequence reduces to the correct
    normal form. The existence of a natural number $n$ with the prescribed properties can be proved
    by observing that none of the reductions above have a length which depends on the parameters
    $\TRSonetolambda{\termone_1},\ldots,\TRSonetolambda{\termone_{m}}$ and $\cltermone_{m+1}\ldots\cltermone_{\arity{\conone_i}}$.
  \end{varitemize}
\end{varitemize}
This concludes the proof.
\end{proof}
So, the required lambda term $\TRSonetolambdaII{\conone_i}$ is simply $\lambdaone^i_{\varepsilon}$.
Interpreting function symbols is more difficult, since we have to ``embed'' the reduction rules
into the $\lambda$-term interpreting the function symbol. To do that, we need a preliminary
result to encode pattern matching.
}
{
The map we are looking for should act compositionally on terms:
\begin{eqnarray*}
\TRSonetolambdaII{\conone(\termone_1,\ldots,\termone_{\arity{\conone}})}&=&\TRSonetolambdaII{\conone}
  \TRSonetolambdaII{\termone_1}\ldots\TRSonetolambdaII{\termone_{\arity{\conone}}}\\
\TRSonetolambdaII{\funone(\termone_1,\ldots,\termone_{\arity{\funone}})}&=&\TRSonetolambdaII{\funone}
  \TRSonetolambdaII{\termone_1}\ldots\TRSonetolambdaII{\termone_{\arity{\funone}}}.
\end{eqnarray*}
As a consequence, we only need to define our map on constructors
and on function symbols. Defining the $\lambda$-term $\TRSonetolambdaII{\conone}$ corresponding to
a constructor $\conone$ is relatively easy: we simply need to take into account the
case when some argument is $\errorterm$. Treating function symbols is
more complicated, since the rewrite rules governing $\funone$ must be ``embedded'' inside
$\TRSonetolambdaII{\funone}$. One of the main ingredients in rewriting is first-order matching, which is
not natively available in the \emph{pure} $\lambda$-calculus, and should therefore
be coded into the $\lambda$-term. This is the purpose of the 
following lemma. 
}
\begin{lemma}[Pattern matching]\label{lemma:pm}
Let $\seqpone_1,\ldots,\seqpone_n$ be non-overlapping sequences of patterns of the same length $m$.
Then there are a term $\lambdaone_{\seqpone_1,\ldots,\seqpone_n}^m$ and an integer $l$ such that
for every sequence of values $\valueone_1,\ldots,\valueone_n$,
if $\seqpone_i=\patone_1,\ldots,\patone_m$ then
\begin{equation*}
\begin{split}
\lambdaone_{\seqpone_1,\ldots,\seqpone_n}^m
  \TRSonetolambda{\patone_1(\termone_1^1,\ldots,\termone_1^{k_1})}
  \ldots
  \TRSonetolambda{\patone_m(\termone_m^1,\ldots,\termone_m^{k_m})}
   \valueone_1\ldots \valueone_n \\ 
\rewrlambdav^k \; 
\valueone_i\TRSonetolambda{\termone_1^1}\ldots\TRSonetolambda{\termone_1^{k_1}}
\ldots\TRSonetolambda{\termone_m^1}\ldots\TRSonetolambda{\termone_m^{k_m}},
\end{split}
\end{equation*}
where $k\leq l$, whenever the $\termone_i^j$ are constructor terms. Moreover,
$$
\lambdaone_{\seqpone_1,\ldots,\seqpone_n}^m\cltermone_1,\ldots,\cltermone_mV_1\ldots V_n
\rewrlambdav^k \errorterm,
$$
where $k\leq l$,
whenever $\cltermone_1,\ldots,\cltermone_m$ do not unify with
any of the sequences $\seqpone_1,\ldots,\seqpone_n$ or any of 
the $\cltermone_1,\ldots,\cltermone_m$ is itself $\errorterm$.
\end{lemma}
\condinc{
\begin{proof}
We go by induction on $p=\sum_{i=1}^n\dsize{\seqpone_i}$, where
$\dsize{\seqpone_i}$ is the number of constructors occurrences in patterns
inside $\seqpone_i$:
\begin{varitemize}
\item
  If $p=0$ and $n=0$, then we should always return $\errorterm$: 
  $$
  \lambdaone_{\varepsilon}^m\equiv\lambda\varone_1.\ldots.\lambda\varone_m.\errorterm.
  $$
\item
  If $p=0$ and $n=1$ and $\seqpone_1$ is simply a sequence of variables 
  $\varone_1,\ldots,\varone_m$ (because the $\seqpone_i$ are assuming to be non-overlapping).
  Then $\lambdaone_{\varone_1,\ldots,\varone_m}^m$ is a term defined by induction on $m$
  which returns $\errorterm$ only if one of its first $m$ arguments is $\errorterm$ and otherwise
  returns its $m+1$-th argument applied to its first $m$ arguments.
\item
  If $p\geq 1$, then there must be integers $i$ and $j$ with 
  $1\leq i\leq m$ and $1\leq j\leq n$ such that
  $$
  \seqpone_j=\patone_1,\ldots,\patone_{i-1},\conone_k(\pattwo_1,\ldots,\pattwo_{\arity{\conone_k}}),
  \patone_{i+1},\ldots,\patone_{m}
  $$
  for a constructor $\conone_k$ and for some patterns $\patone_p$ and some $\pattwo_q$.
  Now, for every $1\leq p\leq n$ 
  and for every $1\leq j\leq g$ 
  we define sequences of patterns $\seqptwo_p^j$ and values $\valuetwo_p^j$ as follows:
  \begin{varitemize}
  \item
    If 
    $$
    \seqpone_p=\patone_1,\ldots,\patone_{i-1},\conone_j(\pattwo_1,\ldots,\pattwo_{\arity{\conone_j}}),
    \patone_{i+1}\ldots\patone_m
    $$
    then $\seqptwo_p^j$ is defined to be the sequence
    $$
    \patone_1,\ldots,\patone_{i-1},\pattwo_1,\ldots,\pattwo_{\arity{\conone_k}},\patone_{i+1},\ldots,\patone_m.
    $$
    Moreover, $\valuetwo_p$ is simply the indentity $\lambda\varone.\varone$.
  \item
    If 
    $$
    \seqpone_p=\patone_1,\ldots,\patone_{i-1},\conone_s(\pattwo_1,\ldots,\pattwo_{\arity{\conone_s}}),
    \patone_{i+1}\ldots\patone_m
    $$
    where $s\neq j$ then $\seqptwo_p^j$ and $\valuetwo_p^j$ are both undefined.
  \item
    Finally, if
    $$
    \seqpone_p=\patone_1,\ldots,\patone_{i-1},\varone,\patone_{i+1}\ldots\patone_m
    $$
    then $\seqptwo_p^j$ is defined to be the sequence
    $$
    \patone_1,\ldots,\patone_{i-1},\varone_1,\ldots,\varone_{\arity{\conone_j}},\patone_{i+1},\ldots,\patone_m.
    $$
    and $\valuetwo_p^j$ is the following $\lambda$-term
    $$
    \lambda\varone.\lambda\vartwo_1.\ldots.\lambda\vartwo_t.\varone_1.\ldots.\lambda\varone_{\arity{\conone_k}}.
    \lambda\varthree_1.\ldots.\lambda\varthree_u.\varone\vartwo_1\ldots\vartwo_t
    (\TRSonetolambda{\conone_j}\varone_1\ldots\varone_{\arity{\conone_j}})\varthree_1\ldots\varthree_u
    $$
    where $t$ is the number of variables in $\patone_1,\ldots,\patone_{i-1}$ and $u$ is the
    number of variables in $\patone_{i+1},\ldots,\patone_{m}$.
  \end{varitemize}
  As a consequence, for every $1\leq j\leq g$, we can find a natural number $t_j$ and 
  a sequence of pairwise distinct natural numbers $i_1,\ldots,i_{t_j}$ such that 
  $\seqptwo^j_{i_1},\ldots,\seqptwo^j_{i_{t_j}}$ are exactly the sequences which can
  be defined by the above construction. We are now able to formally define
  $\lambdaone_{\seqpone_1,\ldots,\seqpone_n}^m$; it is the term
  $$
  \lambda\varone_1.\ldots.\lambda\varone_m.\lambda\vartwo_1.\ldots.\lambda\vartwo_n.
  ((\varone_i\valueone_1\ldots\valueone_g\valueone_\bot)\varone_1\ldots\varone_{i-1}\varone_{i+1}\ldots\varone_m)
  \vartwo_1\ldots\vartwo_n
  $$
  where
  \begin{eqnarray*}
    \forall 1\leq j\leq g.\valueone_j&\equiv&\lambda\varthree_1.\ldots.\lambda\varthree_{\arity{\conone_j}}.
    \lambda\varone_1.\ldots.\lambda\varone_{i-1}.\lambda\varone_{i+1}.\ldots.\lambda\varone_m.
    \lambda\vartwo_1.\ldots.\lambda\vartwo_n.\\
    & &\lambdaone_{\seqptwo^j_{i_1},\ldots,\seqptwo^j_{i_{t_j}}}^{m-1+\arity{\conone_j}}\varone_1\ldots\varone_{i-1}\varthree_1\ldots\varthree_{\arity{\conone_j}}
      \varone_{i+1}\ldots\varone_m(\valuetwo^j_{i_1}\vartwo_{i_1})\ldots(\valuetwo^j_{i_{t_j}}\vartwo_{i_{t_j}})\\
    \valueone_\bot&\equiv&\lambda\varone_1.\ldots.\lambda\varone_{i-1}.\lambda\varone_{i+1}.\ldots.\lambda\varone_m.
    \lambda\vartwo_1.\ldots.\lambda\vartwo_n.\errorterm
  \end{eqnarray*}
  Notice that, for every $j$,
  $p>\sum_{v=1}^{t_j}\dsize{\seqptwo^j_v}$. Moreover, for every $j$ any
  $\seqptwo^j_v$ has the same length $m-1+\arity{\conone_j}$.
  This justifies the application of the
  induction hypothesis above.
\end{varitemize}
This concludes the proof.
\end{proof}
For every function symbol $\funone_i$, let
$$
\funone_i(\seqpone_i^1)\rewrTRS\termone_i^1,\ldots,
\funone_i(\seqpone_i^{n_i})\rewrTRS\termone_i^{n_i}
$$
be the rules for $\funone_i$. Moreover, suppose that
the variables appearing in the patterns in $\seqpone_i^j$ are
$z_i^{j,1},\ldots,z_i^{j,m_{i,j}}$.
Recall that we have a signature with $\funone_1,\ldots,\funone_h$ 
function symbols. For any $1\leq i \leq h$ 
the lambda term interepreting $\funone_i$ is defined to be:
$$
\TRSonetolambdaII{\funone_i}\equiv H_i\valueone_1\ldots\valueone_h
$$
where
\begin{eqnarray*}
\valueone_i&\equiv&\lambda x_1.\ldots.\lambda x_h.\lambda y_1.\ldots.\lambda y_{\arity{\funone_i}}.
\lambdaone_{\seqpone_i^1,\ldots,\seqpone_i^n}y_1\ldots y_{\arity{\funone_i}}\valuetwo_i^1\ldots\valuetwo_i^{n_i}\\
\valuetwo_i^j&\equiv&\lambda\varthree_1.\ldots.\lambda\varthree_{m_{i,j}}.\TRSonetolambdaIII{\termone_i^j}
\end{eqnarray*}
whenever $1\leq i\leq h$ and $1\leq j\leq n_i$. Moreover
$\TRSonetolambdaIII{\cdot}$ is defined by induction as follows:
\begin{eqnarray*}
\TRSonetolambdaIII{\varone}&=&\varone\\
\TRSonetolambdaIII{\conone_i(\termone_1,\ldots,\termone_{\arity{\conone_i}})}&=&\TRSonetolambdaII{\conone_i}
  \TRSonetolambdaIII{\termone_1}\ldots\TRSonetolambdaIII{\termone_{\arity{\conone_i}}}\\
\TRSonetolambdaIII{\funone_i(\termone_1,\ldots,\termone_{\arity{\funone_i}})}&=&\varone_i
  \TRSonetolambdaIII{\termone_1}\ldots\TRSonetolambdaIII{\termone_{\arity{\funone_i}}}
\end{eqnarray*}
Now, we have all the necessary ingredients to extend the mapping $\TRSonetolambdaII{\cdot}$
to every term in $\TRStermsp{\TRSone}$:
\begin{eqnarray*}
\TRSonetolambdaII{\conone(\termone_1,\ldots,\termone_{\arity{\conone_i}})}&=&\TRSonetolambdaII{\conone_i}
  \TRSonetolambdaII{\termone_1}\ldots\TRSonetolambdaII{\termone_{\arity{\conone_i}}}\\
\TRSonetolambdaII{\funone_i(\termone_1,\ldots,\termone_{\arity{\funone_i}})}&=&\TRSonetolambdaII{\funone_i}
  \TRSonetolambdaII{\termone_1}\ldots\TRSonetolambdaII{\termone_{\arity{\funone_i}}}
\end{eqnarray*}
}
{
Rewrite rules define the computational behavior of function symbols in a mutually recursive
way. We can exploit the powerful fixed-point combinators of the untyped $\lambda$-calculus
(see Lemma~\ref{lemma:mfpc}) to define 
$\TRSonetolambdaII{\funone}$. 
}

\begin{theorem}\label{theo:simulcl}
There is a natural number $k$ such that for every
function symbol $\funone$ and for every 
$\termone_1,\ldots,\termone_{\arity{\funone}}\in\TRScontermsp{\TRSone}$,
the following three implications hold (where
$\termtwo$ stands for $\funone(\termone_1,\ldots,\termone_{\arity{\funone}})$
and $\lambdaone$ stands for $\TRSonetolambdaII{\funone}\TRSonetolambda{\termone_1}\ldots\TRSonetolambda{\termone_{\arity{\funone}}}$):
\begin{varitemize}
\item
If $\termtwo$ rewrites to
$\termthree\in\TRScontermsp{\TRSone}$ in $n$ steps, then
$\lambdaone$ rewrites to $\TRSonetolambda{\termthree}$ in
at most $kn$ steps.
\item
If $\termtwo$ rewrites to a normal form
$\termthree\notin\TRScontermsp{\TRSone}$, then
$\lambdaone$ rewrites to $\errorterm$.
\item
If $\termtwo$ diverges, then $\lambdaone$ diverges.
\end{varitemize}
\end{theorem}
\condinc{
\begin{proof}
By an easy combinatorial argument following from the definition
of $\TRSonetolambdaII{\cdot}$.
\end{proof}}{
Informally, any term $\TRSonetolambdaII{\funone_i}$
has the form $H_iV_1\ldots V_h$ where 
$$
V_i=\lambda\varone_1.\ldots.\lambda\varone_h.\lambda\vartwo_1.\ldots.\lambda\vartwo_{\arity{\funone_i}}.
\lambdaone_{\seqpone_1,\ldots,\seqpone_n}\vartwo_1\ldots\vartwo_{\arity{\funone_i}}\valuetwo_1\ldots\valuetwo_n,
$$
$n$ is the number of reduction rules ``governing'' $\funone_i$, $\seqpone_i$ is the sequence
of patterns appearing in the $i$-th such rule and $\valuetwo_i$ corresponds
to the rhs of the same rule. As a consequence, the natural number $k$ of Theorem~\ref{theo:simulcl} 
can be obtained from the corresponding constants from Lemma~\ref{lemma:mfpc} and
Lemma~\ref{lemma:pm}.
}
Clearly, the constant $k$ in
Theorem~\ref{theo:simulcl} depends on $\TRSone$, but is independent on the particular
term $\termtwo$.
\condinc{
\section{Graph Representation}
\label{Sect:GraphRep}
The previous two sections proved the main simulation result of the paper.
To complete the picture, we show in this section that the unitary cost model
for the (weak call-by-value) $\lambda$-calculus (and hence the number of
rewriting in a constructor term rewriting system) is polynomially related to the actual cost of implementing those reductions\footnote{As mentioned in the introduction, see~\cite{Sands:Lambda02} for another proof of this with other means.}. We do so by introducing term graph rewriting, following~\cite{TGRbarendregt} but
adapting the framework to call-by-value constructor rewriting. Contrarily to what
we did in Section~\ref{sect:CRS}, we will stay abstract here: our attention will
not be restricted to the particular graph rewrite system that is needed to implement 
reduction in the $\lambda$-calculus. 

We refer the reader to our~\cite{DLM09} for more details on efficient simulations between term graph rewriting and constructor term rewriting, both under innermost (i.e., call-by-value) and outermost (i.e., call-by-name) reduction strategies.

\begin{definition}[Labelled Graph]
Given a signature $\sigone$, a \emph{labelled graph over $\sigone$} consists of a directed
acyclic graph together with an ordering on the outgoing edges of each node and a (partial)
labelling of nodes with symbols from $\sigone$ such that the out-degree of each node
matches the arity of the corresponding symbols (and is $0$ if the labelling is undefined).
Formally, a labelled graph is a triple $\tgone=(\vsone,\ordone,\labelone)$
where: 
\begin{varitemize}
\item
  $\vsone$ is a set of \emph{vertices}.
\item
  $\ordone:\vsone\rightarrow\vsone^*$ is a (total) \emph{ordering function}.
\item
  $\labelone:\vsone\rightharpoonup\vsone$ is a (partial) \emph{labelling function} such
  that the length of $\ordone(\verone)$ is the arity of $\labelone(\verone)$ if
  $\labelone(\verone)$ is defined and is $0$ otherwise.
\end{varitemize}
A labelled graph $(\vsone,\ordone,\labelone)$ is \emph{closed} iff $\labelone$ is a 
total function. 
\end{definition}
Consider the signature $\Sigma=\{a,b,c,d\}$, where
arities of $a,b,c,d$ are $2$, $1$, $0$, $2$ respectively, and
$b$, $c$, $d$ are constructors. Examples of labelled graphs over 
the signature $\Sigma$ are the following ones:
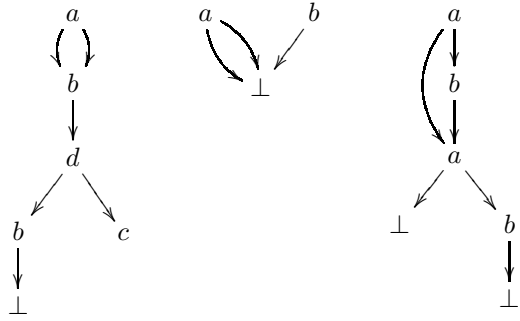
\begin{displaymath}
\xymatrix@R=15pt@C=8pt{
     & a \ar@/_/[d]\ar@/^/[d] &   \\
     & b \ar[d] &   \\
     & d \ar[dl]\ar[dr] &   \\
b \ar[d] &          & c \\
\bot &          &   \\
}
\hspace{20pt}
\xymatrix@R=15pt@C=8pt{
a\ar@/^/[dr]\ar@/_/[dr] & & b \ar[dl]\\
          & \bot &                         \\
}
\hspace{20pt}
\xymatrix@R=15pt@C=8pt{
     & a\ar@/_1pc/[dd]\ar[d] &   \\
     & b \ar[d] &   \\
     & a \ar[dl]\ar[dr] &   \\
\bot &          & b\ar[d] \\
     &          & \bot  \\
}
\end{displaymath}
The symbol $\bot$ denotes vertices where the underlying labelling function
is undefined (and, as a consequence, no edge departs from such vertices).
Their role is similar to the one of variables in terms.

If one of the vertices of a labelled graph is selected as the \emph{root}, we obtain
a term graph:
\begin{definition}[Term Graphs]
A \emph{term graph}, is a quadruple 
$\tgone=(\vsone,\ordone,\labelone,\rootone)$, where 
$(\vsone,\ordone,\labelone)$ is a labelled graph and
$\rootone\in\vsone$ is the \emph{root} of the term graph.
\end{definition}
The following are graphic representations of some term graphs.
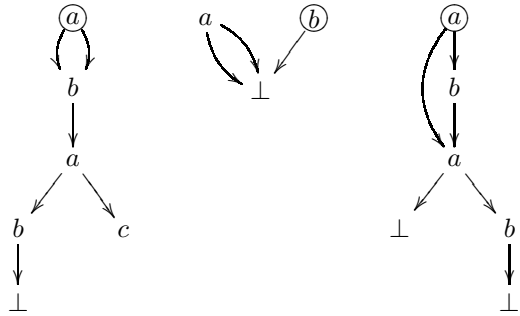
\begin{displaymath}
\xymatrix@R=15pt@C=8pt{
     &  *+[o][F]{a} \ar@/_/[d]\ar@/^/[d] &   \\
     & b \ar[d] &   \\
     & a \ar[dl]\ar[dr] &   \\
b \ar[d] &          & c \\
\bot &          &   \\
}
\hspace{20pt}
\xymatrix@R=15pt@C=8pt{
 a\ar@/^/[dr]\ar@/_/[dr] & &  *+[o][F]{b} \ar[dl]\\
          & \bot &                         \\
}
\hspace{20pt}
\xymatrix@R=15pt@C=8pt{
     & *+[o][F]{a}\ar@/_1pc/[dd]\ar[d] &   \\
     & b \ar[d] &   \\
     & a \ar[dl]\ar[dr] &   \\
\bot &          & b\ar[d] \\
     &          & \bot  \\
}
\end{displaymath}
The root is the only vertex drawn inside a circle.

There are some classes of paths which are particularly relevant for our purposes
\begin{definition}[Paths]
A path $\verone_1,\ldots,\verone_n$ in a labelled graph
$\tgone=(\vsone,\ordone,\labelone)$ is said to be:
\begin{varitemize}
\item
  A \emph{constructor path} iff for every $1\leq i\leq n$, the symbol
  $\labelone(\verone_i)$ is a constructor;
\item
  A \emph{pattern path} iff for every $1\leq i\leq n$, 
  $\labelone(\verone_i)$ is either a constructor symbol or is undefined;
\item
  A \emph{left path} iff $n\geq 1$, the symbol $\labelone(\verone_1)$ is 
  a function symbol and $\verone_2,\ldots,\verone_n$
  is a pattern path.
\end{varitemize}
\end{definition}

\begin{definition}[Homomorphisms]
An \emph{homomorphism} between two labelled graphs 
$\tgone=(\vsone_\tgone,\ordone_\tgone,\labelone_\tgone)$ and 
$\tgtwo=(\vsone_\tgtwo,\ordone_\tgtwo,\labelone_\tgtwo)$ over the same 
signature $\sigone$ is a function $\homone$ from $\vsone_\tgone$ to $\vsone_\tgtwo$ 
preserving the term graph structure. In particular
\begin{eqnarray*}
  \labelone_\tgtwo(\homone(\verone))&=&\labelone_\tgone(\verone)\\
  \ordone_\tgtwo(\homone(\verone))&=&\homone^*(\ordone_\tgone(\verone))
\end{eqnarray*}
for any $\verone\in\domain{\labelone}$, where $\homone^*$ is the obvious
generalization of $\homone$ to sequences of vertices. 
An \emph{homomorphism} between two term graphs 
$\tgone=(\vsone_\tgone,\ordone_\tgone,\labelone_\tgone,\rootone_\tgone)$ and 
$\tgtwo=(\vsone_\tgtwo,\ordone_\tgtwo,\labelone_\tgtwo,\rootone_\tgtwo)$ is
an homomorphism between $(\vsone_\tgone,\ordone_\tgone,\labelone_\tgone)$ and 
$(\vsone_\tgtwo,\ordone_\tgtwo,\labelone_\tgtwo)$ such that
$\homone(\rootone_\tgone)=\rootone_\tgtwo$. Two labelled graphs $\tgone$ and $\tgtwo$ are 
isomorphic iff there is a bijective homomorphism from
$\tgone$ to $\tgtwo$; in this case, we write $\tgone\cong\tgtwo$. Similarly
for term graphs.
\end{definition}
In the following, we will consider term graphs modulo isomorphism, i.e., $\tgone=\tgtwo$
iff $\tgone\cong\tgtwo$. Observe that two isomorphic term graphs have the same graphical
representation.
\begin{definition}[Graph Rewrite Rules]
A \emph{graph rewrite rule} over a signature $\sigone$ 
is a triple $\rrone=(\tgone,\rootone,\roottwo)$ such that:
\begin{varitemize}
\item
  $\tgone$ is a labelled graph;
\item
  $\rootone,\roottwo$ are vertices of $\tgone$, called 
  the \emph{left root} and the \emph{right root} of $\rrone$,
  respectively.
\item
  Any path starting in $\rootone$ is a
  left path.
\end{varitemize}
\end{definition}
The following are examples of graph rewriting rules, assuming $a$ to be a function
symbol and $b,c,d$ to be constructors:
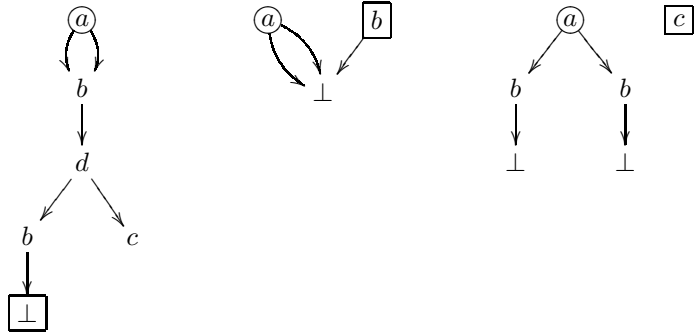
\begin{displaymath}
\xymatrix@R=15pt@C=8pt{
     &  *+[o][F]{a} \ar@/_/[d]\ar@/^/[d] &   \\
     & b \ar[d] &   \\
     & d \ar[dl]\ar[dr] &   \\
b \ar[d] &          & c \\
*+[F]{\bot} &          &   \\
}
\hspace{40pt}
\xymatrix@R=15pt@C=8pt{
 *+[o][F]{a}\ar@/^/[dr]\ar@/_/[dr] & &  *+[F]{b} \ar[dl]\\
          & \bot &                         \\
}
\hspace{40pt}
\xymatrix@R=15pt@C=8pt{
        & *+[o][F]{a}\ar[dl]\ar[dr] &  &  *+[F]{c}\\
b\ar[d] &                  & b\ar[d] &   \\
\bot    &                  & \bot    &   \\
}
\end{displaymath}
\begin{definition}[Subgraphs]
Given a labelled graph $\tgone=(\vsone_\tgone,\ordone_\tgone,\labelone_\tgone)$ 
and any vertex $\verone\in\vsone_\tgone$, the \emph{subgraph of $\tgone$ rooted
at $\verone$}, denoted $\subgr{\tgone}{\verone}$, is the
term graph $(\vsone_{\subgr{\tgone}{\verone}},\ordone_{\subgr{\tgone}{\verone}},\labelone_{\subgr{\tgone}{\verone}},
\rootone_{\subgr{\tgone}{\verone}})$ where
\begin{varitemize}
\item
  $\vsone_{\subgr{\tgone}{\verone}}$ is the subset of $\vsone_\tgone$ whose elements
  are vertices which are reachable from $\verone$ in $\tgone$.
\item
  $\ordone_{\subgr{\tgone}{\verone}}$ and $\labelone_{\subgr{\tgone}{\verone}}$ are
  the appropriate restrictions of $\ordone_\tgone$ and $\labelone_\tgone$
  to $\vsone_{\subgr{\tgone}{\verone}}$.
\item
  $\rootone_{\subgr{\tgone}{\verone}}$ is $\verone$.
\end{varitemize}
\end{definition}
\begin{definition}[Redexes]
Given a labelled graph $\tgone$, a \emph{redex} for $\tgone$ is
a pair $(\rrone,\homone)$, where $\rrone$ is a rewrite rule 
$(\tgtwo,\rootone,\roottwo)$ and $\homone$ is an homomorphism
between $\subgr{\tgtwo}{\rootone}$ and $\tgone$
such that for any vertex $\verone\in\vsone_{\subgr{\tgtwo}{\rootone}}$
with $\verone\notin\domain{\labelone_{\subgr{\tgtwo}{\rootone}}}$, 
any path starting in $\homone(\verone)$ is a constructor path.
\end{definition}
The last condition in the definition of a redex is needed to
capture the call-by-value nature of the rewriting process.

Given a term graph $\tgone$ and a redex $((\tgtwo,\rootone,\roottwo),\homone)$,
the result of firing the redex is another term graph obtained by
successively applying the following three steps to $\tgone$:
\begin{varenumerate}
\item
  The \emph{build phase}: create an isomorphic copy of the portion of
  $\subgr{\tgtwo}{\roottwo}$ not contained in
  $\subgr{\tgtwo}{\rootone}$, and add it to $\tgone$, obtaining $\tgthree$.
  The underlying ordering and labelling functions are defined in the natural
  way.
\item
  The \emph{redirection phase}: all edges in $\tgthree$ pointing to $\homone(\rootone)$
  are replaced by edges pointing to the copy of $\roottwo$. If $\homone(\rootone)$
  is the root of $\tgone$, then the root of the newly created graph will be
  the newly created copy of $\roottwo$. The graph $\tgfour$ is obtained.
\item
  The \emph{garbage collection phase}: all vertices which are not accessible
  from the root of $\tgfour$ are removed. The graph $\tgfive$ is obtained.
\end{varenumerate}
We will write $\tgone\rewrite{(\tgtwo,\rootone,\roottwo)}\tgfive$ (or simply
$\tgone\rewrgraph\tgfive$, if this does not cause ambiguity) in this case.

As an example, consider the term graph $\tgone$ and the rewriting rule 
$\rrone=(\tgtwo,\rootone,\roottwo)$:
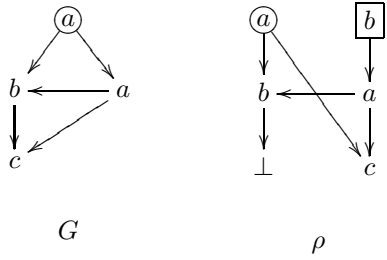
\begin{displaymath}
\xymatrix@R=15pt@C=8pt{
  &  *+[o][F]{a}\ar[dl]\ar[dr] &   \\
  b \ar[d] & & a \ar[ll]\ar[dll] \\
  c & &  \\
  & \tgone & \\
}
\hspace{40pt}
\xymatrix@R=15pt@C=8pt{
 *+[o][F]{a} \ar[d]\ar[ddrr] & & *+[F]{b}\ar[d] \\
 b \ar[d] & & a\ar[ll]\ar[d] \\
\bot & & c \\
 & \rrone & \\
}
\end{displaymath}
There is an homomorphism $\homone$ from
$\subgr{\tgtwo}{\rootone}$ to
$\tgone$. In particular, $\homone$ maps $\rootone$ to
the rightmost vertex in $\tgone$.
Applying the build phase and the redirection phase we get $\tgthree$
and $\tgfour$ as follows:
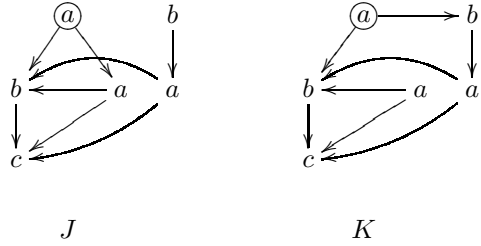
\begin{displaymath}
\xymatrix@R=15pt@C=8pt{
  &  *+[o][F]{a}\ar[dl]\ar[dr] & & b\ar[d]  \\
  b \ar[d] & & a \ar[ll]\ar[dll] & a\ar@/^/[dlll]\ar@/_1pc/[lll] \\
  c & & & \\
  & \tgthree & & \\
}
\hspace{40pt}
\xymatrix@R=15pt@C=8pt{
  &  *+[o][F]{a}\ar[dl]\ar[rr] & & b\ar[d]  \\
  b \ar[d] & & a \ar[ll]\ar[dll] & a\ar@/^/[dlll]\ar@/_1pc/[lll] \\
  c & & & \\
  & \tgfour & & \\
}
\end{displaymath}
Finally, applying the garbage collection phase, we get the
result of firing the redex $(\rrone,\homone)$:
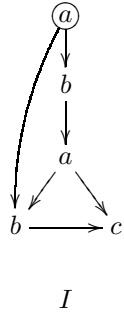
\begin{displaymath}
\xymatrix@R=15pt@C=8pt{
 & *+[o][F]{a}\ar@/_/[dddl]\ar[d] & \\
 & b\ar[d] & \\
 & a\ar[dl]\ar[dr] & \\
 b\ar[rr] & & c\\
 & I & \\
}
\end{displaymath}

\begin{definition}
A constructor graph rewrite system (CGRS) over a signature $\sigone$ consists of
a set of graph rewrite rules $\sgrone$ on $\sigone$.
\end{definition}
\subsection{From Term Rewriting to Graph Rewriting}
Any term $\termone$ over the signature $\sigone$ can be turned into a
graph $\tgone$ in the obvious way: $\tgone$ will be
a tree and vertices in $\tgone$ will be in one-to-one correspondence
with symbol occurrences in $\termone$. Conversely, any term graph
$\tgone$ over $\sigone$ can be turned into a term $\termone$
over $\sigone$ (remember: we only consider acyclic graphs here).
Similarly, any term rewrite rule $\termone\rewrTRS\termtwo$ 
over the signature $\sigone$ can
be translated into a graph rewrite rule 
$(\tgone,\rootone,\roottwo)$
as follows:
\begin{varitemize}
\item
  Take the graph representing $\termone$ and $\termtwo$. They are
  trees, in fact. 
\item
  From the union of these two trees, share those
  nodes representing the same variable in $\termone$ and $\termtwo$.
  This is $\tgone$.
\item
  Take $\rootone$ to be the root of $\termone$ in $\tgone$ and
  $\roottwo$ to be the root of $\termtwo$ in $\tgone$.
\end{varitemize}

As an example, consider the rewriting rule
$$
a(b(x),y)\rewrTRS b(a(y,a(y,x))).
$$
Its translation as a graph rewrite rule is the following:
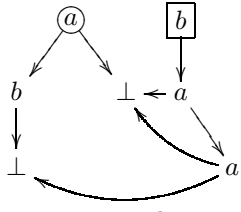
\begin{displaymath}
\xymatrix@R=15pt@C=8pt{
 & *+[o][F]{a}\ar[dl]\ar[dr] & & *+[F]{b}\ar[d] & \\
 b\ar[d] & & \bot & a\ar[l]\ar[dr] & \\
 \bot & & & & a\ar@/^/[ull]\ar@/^1pc/[llll]\\
}
\end{displaymath}

An arbitrary constructor rewriting system can be turned 
into a constructor graph rewriting system:
\begin{definition}
Given a constructor rewriting system $\srrone$ over $\sigone$, the 
corresponding constructor graph rewriting system $\sgrone$ is defined 
as the class of graph rewrite rules corresponding to those in $\srrone$.
Given a term $\termone$, $\CtoCG{\termone}$ will be the corresponding
graph, while the term graph $\tgone$ corresponds to the term $\CGtoC{\tgone}$.
\end{definition}

Let us now consider graph rewrite rules corresponding to
rewrite rules in $\srrone$. It is easy to realize that
the following invariant is preserved while performing
rewriting in $\CtoCG{\srrone}$: whenever any vertex $\verone$ can
be reached by two distinct paths starting at the root
(i.e., $\verone$ is \emph{shared}), any path starting at
$\verone$ is a constructor path. A term graph
satisfying this invariant is said to be 
\emph{constructor-shared}.

Constructor-sharedness holds for term graphs coming from terms and
is preserved by graph rewriting:
\begin{lemma}\label{lemma:constructorsharedness}
For every closed term $\termone$, $\CtoCG{\termone}$ is constructor-shared.
Moreover, if $\tgone$ is closed and constructor-shared and $\tgone\rewrgraph\tgfive$, 
then $\tgfive$ is constructor-shared.
\end{lemma}
\begin{proof}
The fact $\CtoCG{\termone}$ is constructor-shared for every $\termone$ follows
from the way the $\CtoCG{\cdot}$ map is defined: it does not introduce any
sharing. Now, suppose $\tgone$ is constructor-shared and
$$
\tgone\rewrite{(\tgtwo,\rootone,\roottwo)}\tgfive
$$
where $(\tgtwo,\rootone,\roottwo)$ corresponds to a term rewrite
rule $\termone\rewrTRS\termtwo$. The term graph $\tgthree$ obtained from
$\tgone$ by the build phase is itself constructor-shared: it is obtained from
$\tgone$ by adding some new nodes, namely an isomorphic copy of the portion
of $\subgr{\tgtwo}{\roottwo}$ not contained in $\subgr{\tgtwo}{\rootone}$. Notice
that $\tgthree$ is constructor-shared in a stronger sense: any vertex which
can be reached from the newly created copy of $\roottwo$ by two distinct paths
must be a constructor path. This is a consequence of $(\tgtwo,\rootone,\roottwo)$ 
being a graph rewrite rule corresponding to a term rewrite
rule $\termone\rewrTRS\termtwo$, where the only shared vertices are those
where the labelling function is undefined. The redirection phase preserves
itself constructor-sharedness, because only one pointer is redirected
(the vertex is labelled by a function symbol) and the destination 
of this redirection is a vertex (the newly created
copy of $\roottwo$) which had no edge incident to it. Clearly, the
garbage collection phase preserve constructor-sharedness.
\end{proof}
\begin{lemma}\label{lemma:NFconstructorshared}
A closed term graph $\tgone$ in $\sgrone$ is a 
normal form iff $\CGtoC{\tgone}$ is a normal form.
\end{lemma}
\begin{proof}
Clearly, if a closed term graph $\tgone$ is in normal form,
then $\CGtoC{\tgone}$ is a term in normal form, because each
redex in $\tgone$ translates to a redex in $\CGtoC{\tgone}$.
On the other hand, if $\CGtoC{\tgone}$ is in normal form,
then $\tgone$ is in normal form: each redex in $\CGtoC{\tgone}$
translates back to a redex in $\tgone$.
\end{proof}
Reduction at the level of graphs correctly simulates reduction
at the level of terms, but only if the underlying graphs
are constructor shared:
\begin{lemma}\label{lemma:CGtoC}
If $\tgone$ is closed and constructor-shared and 
$\tgone\rewrgraph\tgfive$, then
$\CGtoC{\tgone}\rewrgraph\CGtoC{\tgfive}$.
\end{lemma}
\begin{proof}
The fact each reduction step starting in $\tgone$ can be
mimicked by $n$ reduction steps in $\CGtoC{\tgone}$ is known
from the literature. If $\tgone$ is constructor-shared,
then $n=1$, because any redex in a constructor-shared
term graph cannot be shared.
\end{proof}
A counterexample, when $\tgone$ in not constructor-shared
can be easily built: consider the term rewrite rule
$a(c,c)\rewrTRS c$ and the following term graph, which is not
constructor-shared and correspond to $a(a(c,c),a(c,c))$:
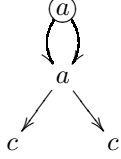
\begin{displaymath}
\xymatrix@R=15pt@C=8pt{
 & *+[o][F]{a}\ar@/_/[d]\ar@/^/[d] &  \\
 & a\ar[dl]\ar[dr] & \\
 c & & c\\
}
\end{displaymath}
The term graph rewrites in \emph{one} step to the following one
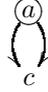
\begin{displaymath}
\xymatrix@R=15pt@C=8pt{
 & *+[o][F]{a}\ar@/_/[d]\ar@/^/[d] &  \\
 & c & \\
}
\end{displaymath}
while the term $a(a(c,c),a(c,c))$ rewrites to $a(c,c)$ in
\emph{two} steps.

As can be expected, graph reduction is even complete with respect to
term reduction, with the only \emph{proviso} that term graphs must
be constructor-shared:
\begin{lemma}\label{lemma:CtoCG}
If $\termone\rewrTRS\termtwo$, $\tgone$ is constructor-shared 
and $\CGtoC{\tgone}=\termone$, then
$\tgone\rewrgraph\tgfive$, where $\CGtoC{\tgfive}=\termtwo$.
\end{lemma}

\begin{theorem}[Graph Reducibility]\label{theo:graphreducible}
For every constructor rewrite system $\srrone$ over $\sigone$ and
for every term $\termone$ over $\sigone$, the following two conditions are
equivalent:
\begin{varenumerate}
\item
  $\termone\rewrTRS^n\termtwo$, where $\termtwo$ is in normal form;
\item
  $\CtoCG{\termone}\rewrTRS^n\tgone$, where $\tgone$ is in normal
  form and $\CGtoC{\tgone}=\termtwo$.
\end{varenumerate}
\end{theorem}
\begin{proof}
Suppose $\termone\rewrgraph^n\termtwo$, where $\termtwo$ is in normal form. 
Then, by applying Lemma~\ref{lemma:CtoCG}, we obtain a term graph $\tgone$ such that
$\CtoCG{\termone}\rewrgraph^n\tgone$ and $\CGtoC{\tgone}=\termtwo$.
By Lemma~\ref{lemma:constructorsharedness}, $\tgone$ is canonical and, by 
Lemma~\ref{lemma:NFconstructorshared}, 
it is in normal form. Now, suppose $\CtoCG{\termone}\rewrgraph^n\tgone$ where
$\CGtoC{\tgone}=\termtwo$ and $\tgone$ is in normal form. By
applying $n$ times Lemma~\ref{lemma:CGtoC}, we obtain
that $\CGtoC{\CtoCG{\termone}}\rewrTRS^n\CGtoC{\tgone}=\termtwo$.
But $\CGtoC{\CtoCG{\termone}}=\termone$ and $\termtwo$ is a normal form 
by Lemma~\ref{lemma:NFconstructorshared}, since 
$\CtoCG{\termone}$ and $\tgone$ are
constructor shared due to Lemma~\ref{lemma:constructorsharedness}.
\end{proof}

There are \emph{term} rewrite systems which are not graph reducible, i.e.
for which the two conditions of Theorem~\ref{theo:graphreducible} are
not equivalent (see~\cite{TGRbarendregt}). However, any 
\emph{othogonal constructor} rewrite system is graph reducible, due to the 
strict constraints on the shape of rewrite rules~\cite{Plump90ggacs}.
This result can be considered as a by-product of our analysis, for which graph rewriting
is only instrumental.

\subsection{Lambda-Terms Can Be Efficiently Reduced by Graph Rewriting}
\label{sect:MainResult}
As a corollary of Theorem~\ref{theo:graphreducible} and Theorem~\ref{theo:termreducible},
we obtain the possibility of reducing $\lambda$-terms by term graphs over $\Functions{\TRS}$.
To this purpose, we can use the CGRS $\GRS$ corresponding to $\TRS$:
\begin{corollary}
Let $\lambdaone\in\Lambdaterms$ be a closed term. The following
two conditions are equivalent:
\begin{varenumerate}
\item
  $\lambdaone\rewrlambdav^n\lambdatwo$ where $\lambdatwo$ is in normal form;
\item
  $\TRStoGRS{\LambdatoTRS{\lambdaone}}\rewrTRS^n\tgone$ where
  $\TRSonetolambda{\GRStoTRS{\tgone}}=\lambdatwo$ and $\tgone$ is in normal form.
\end{varenumerate}
\end{corollary} 
However, there are some missing tales. Let us analyze more closely the combinatorics 
of graph rewriting in $\GRS$:
\begin{varitemize}
\item
  Consider a closed $\lambda$-term $\lambdaone$ and a term graph $\tgone$ such
  that $\TRStoGRS{\LambdatoTRS{\lambdaone}}\rewrgraph^*\tgone$. By Proposition~\ref{prop:constred}
  and Lemma~\ref{lemma:CGtoC}, for every constructor $\constr{\varone}{\lambdatwo}$ appearing
  as a label of a vertex in $\tgone$, $\lambdatwo$ is a subterm of $\lambdaone$. 
\item
  As a consequence, if $\TRStoGRS{\LambdatoTRS{\lambdaone}}\rewrgraph^*\tgone\rewrgraph\tgtwo$,
  then the difference $\length{\tgtwo}-\length{\tgone}$ cannot be too big: at most
  $\length{\lambdaone}$. As a consequence, if 
  $\TRStoGRS{\LambdatoTRS{\lambdaone}}\rewrgraph^n\tgone$ then $\length{\tgone}\leq(n+1)\length{\lambdaone}$.
  Here, we exploit in an essential way the possibility of sharing constructors.
\item
  Whenever $\TRStoGRS{\LambdatoTRS{\lambdaone}}\rewrgraph^n\tgone$, computing
  a graph $\tgtwo$ such that $\tgone\rewrgraph\tgtwo$ takes polynomial time in
  $\length{\tgone}$, which is itself polynomially bounded by $n$ and $\length{\lambdaone}$.
\end{varitemize}
Hence:
\begin{theorem}\label{thm:main}
There is a polynomial $p:\N^2\rightarrow\N$ such that for every $\lambda$-term $\lambdaone$,
the normal form of $\TRStoGRS{\LambdatoTRS{\lambdaone}}$ can be computed in time at most 
$p(|\lambdaone|,\Time{\lambdaone})$.
\end{theorem}
As we mentioned in the introduction, this cannot be achieved when using explicit representations
of $\lambda$-terms. Moreover, reading back a $\lambda$-term from a term graph can take exponential
time, as we mentioned in the introduction.

We can complement Theorem~\ref{thm:main} with a completeness statement --- any universal computational model 
with an invariant cost model can be embedded in the $\lambda$-calculus with a polynomial
overhead. We can exploit for this the analogous result we proved in~\cite{CIE2006} (Theorem 1) --- 
the unitary cost model is easily proved to be more parsimonious than 
the difference cost model considered in~\cite{CIE2006}.

\begin{theorem}
Let $f:\Sigma^*\rightarrow\Sigma^*$ be computed by a Turing machine
$\TMone$ in time $g$. Then, there are a $\lambda$-term $\lambdatwo_\TMone$
and a suitable encoding $\cod{\cdot}:\Sigma^*\rightarrow\Lambdaterms$ 
such that $\lambdatwo_\TMone\cod{v}$ normalizes to $\cod{f(v)}$ in 
$O(g(|v|))$ beta steps.
\end{theorem}
}
{
\section{By-product: Polynomial Invariance}\label{Sect:Byproduct}
In this section, we will show how the correspondence between
$\lambda$-calculus and constructor rewrite systems allows to prove an
invariance result on $\lambda$-calculus originally established by
Sands, Gustavsson, and Moran~\cite{Sands:Lambda02}.
We sketch here how we may obtain the same result by using our
simulation results. The interested reader could find all the technical details 
in~\cite{longversion}.

The result we are interested in is the following: for weak call-by-value $\lambda$-calculus,
the number of beta-reductions to normal form is polynomially
related to the actual cost of computing that normal form
on a Turing machine (and thus on any other reasonable
machine model).
Both in constructor
term rewriting and in the $\lambda$-calculus, proving the
number of reduction steps to normal form to be a ``good''
cost model is not trivial. Indeed, the substitution process
could involve the manipulation (in particular, the copying) 
of terms whose size is not even polynomially
related to the size of the original term nor to the number of
reduction steps performed insofar. As a consequence, substitution
must be implemented carefully, and terms have to be represented
implicitly, exploiting sharing of subterms.
The term rewriting community has already developed a formalism
which allows us to exploit sharing, namely term graph rewriting.
The desired result thus follows easily once we observe that 
\emph{every} (orthogonal) constructor rewrite system is
graph reducible (which is the property expressed by Theorem~\ref{theo:graphreducible}, below; recall that this does not hold for
arbitrary term rewriting systems~\cite{TGRbarendregt}).

Given a signature $\sigone$, a \emph{labelled graph over $\sigone$} 
consists of a directed acyclic graph together with an ordering on 
the outgoing edges of each node, and a (partial) labelling of nodes 
with symbols from $\sigone$ such that the out-degree of each node
matches the arity of the corresponding symbols 
(and is $0$ if the labelling is undefined).
As an example, consider the signature $\Sigma=\{a,b,c,d\}$, where
arities of $a,b,c,d$ are $2$, $1$, $0$, $2$ respectively, and
$b$, $c$, $d$ are constructors. Examples of labelled graphs over 
the signature $\Sigma$ are the following ones:
\begin{center}
\begin{minipage}[c]{3.5cm}
\begin{displaymath}
\xymatrix@R=9pt@C=8pt{
     & a \ar@/_/[d]\ar@/^/[d] &   \\
     & b \ar[d] &   \\
     & d \ar[dl]\ar[dr] &   \\
b \ar[d] &          & c \\
\bot &          &   \\
}
\end{displaymath}
\end{minipage}
\begin{minipage}[c]{3.5cm}
\begin{displaymath}
\xymatrix@R=9pt@C=8pt{
a\ar@/^/[dr]\ar@/_/[dr] & & b \ar[dl]\\
          & \bot &                         \\
}
\end{displaymath}
\end{minipage}
\begin{minipage}[c]{3.5cm}
\begin{displaymath}
\xymatrix@R=9pt@C=8pt{
     & a\ar@/_1pc/[dd]\ar[d] &   \\
     & b \ar[d] &   \\
     & a \ar[dl]\ar[dr] &   \\
\bot &          & b\ar[d] \\
     &          & \bot  \\
}
\end{displaymath}
\end{minipage}
\end{center}
The symbol $\bot$ denotes vertices where the underlying labelling function
is undefined (and, as a consequence, no edge departs from such vertices).
Their role is similar to the one of variables in terms.
A \emph{term graph} is a labelled graph
$\tgone$ together with a node $\rootone$ of $\tgone$, the
\emph{root} of the term graph. As
an example, the following are graphic representations of some term graphs
(over the same signature $\sigone$ considered in the previous examples).
\begin{center}
\begin{minipage}[c]{3.5cm}
\begin{displaymath}
\xymatrix@R=9pt@C=8pt{
     &  *+[o][F]{a} \ar@/_/[d]\ar@/^/[d] &   \\
     & b \ar[d] &   \\
     & a \ar[dl]\ar[dr] &   \\
b \ar[d] &          & c \\
\bot &          &   \\
}
\end{displaymath}
\end{minipage}
\begin{minipage}[c]{3.5cm}
\begin{displaymath}
\xymatrix@R=9pt@C=8pt{
 a\ar@/^/[dr]\ar@/_/[dr] & &  *+[o][F]{b} \ar[dl]\\
          & \bot &                         \\
}
\end{displaymath}
\end{minipage}
\begin{minipage}[c]{3.5cm}
\begin{displaymath}
\xymatrix@R=9pt@C=8pt{
     & *+[o][F]{a}\ar@/_1pc/[dd]\ar[d] &   \\
     & b \ar[d] &   \\
     & a \ar[dl]\ar[dr] &   \\
\bot &          & b\ar[d] \\
     &          & \bot  \\
}
\end{displaymath}
\end{minipage}
\end{center}
The root is the only vertex drawn inside a circle. We now need a way
to write programs acting on term graphs.
A \emph{graph rewrite rule} over a signature $\sigone$ 
consists of a labelled graph $\tgone$ together with
two vertices $\rootone,\roottwo$ of $\tgone$. In the
spirit of constructor rewriting, $\rootone$ must be labelled
with a function symbol, while the label of any vertex reachable from $\rootone$ (if any)
must be a constructor symbol. Moreover, every non-labelled vertex reachable from $\roottwo$
can be reached from $\rootone$, too.
The following are examples of graph rewriting rules, assuming $a$ to be a function
symbol and $b,c,d$ to be constructors:
\begin{center}
\begin{minipage}[c]{3.5cm}
\begin{displaymath}
\xymatrix@R=9pt@C=8pt{
     &  *+[o][F]{a} \ar@/_/[d]\ar@/^/[d] &   \\
     & b \ar[d] &   \\
     & d \ar[dl]\ar[dr] &   \\
b \ar[d] &          & c \\
*+[F]{\bot} &          &   \\
}
\end{displaymath}
\end{minipage}
\begin{minipage}[c]{3.5cm}
\begin{displaymath}
\xymatrix@R=9pt@C=8pt{
 *+[o][F]{a}\ar@/^/[dr]\ar@/_/[dr] & &  *+[F]{b} \ar[dl]\\
          & \bot &                         \\
}
\end{displaymath}
\end{minipage}
\begin{minipage}[c]{3.5cm}
\begin{displaymath}
\xymatrix@R=9pt@C=8pt{
        & *+[o][F]{a}\ar[dl]\ar[dr] &  &  *+[F]{c}\\
b\ar[d] &                  & b\ar[d] &   \\
\bot    &                  & \bot    &   \\
}
\end{displaymath}
\end{minipage}
\end{center}
The circled vertex corresponds to $\rootone$, while the squared vertex
corresponds to $\roottwo$. A set of graph rewrite rules over
a signature $\sigone$ defines a constructor graph rewrite
system $\sgrone$. The relation $\rewr{\sgrone}$ is defined in
the natural way (see~\cite{longversion}). As an example, 
if $\sgrone$ includes the following
rewrite rule
\begin{displaymath}
\xymatrix@R=9pt@C=8pt{
 *+[o][F]{a} \ar[d]\ar[ddrr] & & *+[F]{b}\ar[d] \\
 b \ar[d] & & a\ar[ll]\ar[d] \\
\bot & & c \\
}
\end{displaymath}
then
\begin{center}
\begin{minipage}[c]{3cm}
\begin{displaymath}
\xymatrix@R=9pt@C=8pt{
  &  *+[o][F]{a}\ar[dl]\ar[dr] &   \\
  b \ar[d] & & a \ar[ll]\ar[dll] \\
  c & &  \\
}
\end{displaymath}
\end{minipage}
\begin{minipage}[c]{.5cm}
\begin{displaymath}
\rewr{\sgrone}
\end{displaymath}
\end{minipage}
\begin{minipage}[c]{3cm}
\begin{displaymath}
\xymatrix@R=9pt@C=8pt{
 & *+[o][F]{a}\ar@/_/[dddl]\ar[d] & \\
 & b\ar[d] & \\
 & a\ar[dl]\ar[dr] & \\
 b\ar[rr] & & c\\
}
\end{displaymath}
\end{minipage}
\end{center}
Given a constructor rewrite system $\srrone$ over $\sigone$, the 
corresponding constructor graph rewrite system $\sgrone$ is defined 
as the class of graph rewrite rules corresponding to those in $\srrone$.
Given a term $\termone$, $\CtoCG{\termone}$ will be the corresponding
graph, while the term graph $\tgone$ corresponds to the term $\CGtoC{\tgone}$.
Contrarily to what happens in term rewriting, any construtor rewrite system
is graph reducible:

\begin{theorem}[Graph Reducibility]\label{theo:graphreducible}
For every constructor rewrite system $\srrone$ over $\sigone$ and
for every term $\termone$ over $\sigone$, the following two conditions are
equivalent:
\begin{varenumerate}
\item
  $\termone\rewrTRS^n\termtwo$, where $\termtwo$ is in normal form;
\item
  $\CtoCG{\termone}\rewrTRS^n\tgone$, where $\tgone$ is in normal
  form and $\CGtoC{\tgone}=\termtwo$.
\end{varenumerate}
\end{theorem}
As a corollary (and thanks to Proposition~\ref{prop:constred}), we get:
\begin{corollary}\label{cor:main}
There is a polynomial $p:\N^2\rightarrow\N$ such that for every $\lambda$-term $\lambdaone$,
the normal form of $\TRStoGRS{\LambdatoTRS{\lambdaone}}$ can be computed in time at most 
$p(|\lambdaone|,\Time{\lambdaone})$.
\end{corollary}
This cannot be achieved when using explicit representations
of $\lambda$-terms. Moreover, reading back a $\lambda$-term from a term graph may take exponential
time.

\begin{example}
Consider the $\lambda$-terms 
$$
\lambdaone_n\equiv\underbrace{\overline{2}(\overline{2}(\overline{2}(\ldots (\overline{2}}_{\mbox{$n$ times}}(\lambda x.x))\ldots))),
$$
where $\overline{2}$ is the Church numeral for $2$. Their size is proportional to $n$, while the size of their normal
form is exponential in $n$. However, the number of reduction steps leading to the normal form is \emph{linear} in $n$. Altogether, this implies
that the number of reduction steps to normal form cannot be a cost model \emph{if we represent $\lambda$-terms explicitly}, i.e., if we require
the output to be a $\lambda$-term expressed as a string. Observe, on the other hand, that the normal
form $\TRStoGRS{\LambdatoTRS{\lambdaone_n}}$ is a graph whose size is again linear in $n$.
\end{example}

We can complement Corollary~\ref{cor:main} with a completeness statement --- any universal computational model 
with an invariant cost model can be embedded in the $\lambda$-calculus with a polynomial
overhead. We can exploit for this the analogous result we proved in~\cite{CIE2006} (Theorem 1) --- 
the unitary cost model is easily proved to be more parsimonious than 
the difference cost model considered in~\cite{CIE2006}.

\begin{theorem}
Let $f:\Sigma^*\rightarrow\Sigma^*$ be computed by a Turing machine
$\TMone$ in time $g$. Then, there are a $\lambda$-term $\lambdatwo_\TMone$
and a suitable encoding $\cod{\cdot}:\Sigma^*\rightarrow\Lambdaterms$ 
such that $\lambdatwo_\TMone\cod{v}$ normalizes to $\cod{f(v)}$ in 
$O(g(|v|))$ beta steps.
\end{theorem}
}

\condinc{
\section{Variations: Call-by-Name Reduction}
\label{sect:HeadReduction}
Our purpose in this last section is showing that similar techniques can be applied to
call-by-name evaluation of $\lambda$-terms.

In the previous sections, $\lambda$-calculus was endowed with weak call-by-value reduction.
The same technique, however, can be applied to weak call-by-name reduction, as we will sketch in
this section. $\Lambdaterms$ is now endowed with a relation $\rewrlambdah$ defined
as follows:
$$
\begin{array}{ccccc}
  \infer{(\lambda\varone.\lambdaone)\lambdatwo\rewrlambdah\lambdaone\{\lambdatwo/\varone\}}{}
  & & & &
  \infer{\lambdaone\lambdathree\rewrlambdah\lambdatwo\lambdathree}{\lambdaone\rewrlambdah\lambdatwo}
\end{array}
$$
Similarly to the call-by-value case, $\Timew{\lambdaone}$ stands for the number of
reduction steps to the normal form of $\lambdaone$ (if any). Since the relation
$\rewrlambdah$ is deterministic (i.e., functional), $\Timew{\lambdaone}$ is well-defined.

We need another CRS, called $\TRSW$, which is similar to $\TRS$ but designed
to simulate weak call-by-name reduction:
\begin{varitemize}
\item
  The signature $\Functions{\TRSW}$ includes the binary function symbol $\appTRS$ and
  constructor symbols $\constr{\varone}{M}$ for every $\lambdaone\in\Lambdaterms$ and
  every $\varone\in\Variables$, exactly as $\Functions{\TRS}$. Moreover, there is another
  binary constructor symbol $\cappTRSW$. 
  To every term $\lambdaone\in\Lambdaterms$ we can associate
  terms $\LambdatoTRSWaux{\lambdaone},\LambdatoTRSW{\lambdaone}\in\TRSWvarterms$ as follows:
  \begin{eqnarray*}
    \LambdatoTRSWaux{\varone}&=&\varone\\
    \LambdatoTRSWaux{\lambda\varone.\lambdaone}&=&\constr{\varone}{\lambdaone}(\varone_1,\ldots,\varone_{n}),
      \mbox{ where $\FV{\lambda\varone.\lambdaone}=\varone_1,\ldots,\varone_n$}\\
    \LambdatoTRSWaux{\lambdaone\lambdatwo}&=&\cappTRSW(\LambdatoTRSWaux{\lambdaone},\LambdatoTRSWaux{\lambdatwo})\\
    \LambdatoTRSW{\varone}&=&\varone\\
    \LambdatoTRSW{\lambda\varone.\lambdaone}&=&\constr{\varone}{\lambdaone}(\varone_1,\ldots,\varone_{n}),
      \mbox{ where $\FV{\lambda\varone.\lambdaone}=\varone_1,\ldots,\varone_n$}\\
    \LambdatoTRSW{\lambdaone\lambdatwo}&=&\appTRS(\LambdatoTRSW{\lambdaone},\LambdatoTRSWaux{\lambdatwo})
  \end{eqnarray*}
  Notice that $\LambdatoTRSWaux{\cdot}$ maps lambda terms to \emph{constructor} terms, while
  terms obtained via $\LambdatoTRSW{\cdot}$ can contain function symbols.
\item
    The rewrite rules in $\Rules{\TRSW}$ are all the rules in the following form:
    \begin{eqnarray*}
      \appTRS(\constr{\varthree}{\varthree},\cappTRSW(\varfour,\varfive))&\rewrTRSW&\appTRS(\varfour,\varfive)\\
      \appTRS(\constr{\varthree}{\varthree},\constr{\varone}{\lambdaone}(\varone_1,\ldots,\varone_n))&\rewrTRSW&
        \constr{\varone}{\lambdaone}(\varone_1,\ldots,\varone_n)\\
      \appTRS(\constr{\varthree}{\varfour}(\cappTRSW(\varfive,\varsix)),\varseven)&\rewrTRSW&\appTRS(\varfive,\varsix)\\
      \appTRS(\constr{\varthree}{\varfour}(\constr{\varone}{\lambdaone}(\varone_1,\ldots,\varone_n)),\varseven)&\rewrTRSW&
        \constr{\varone}{\lambdaone}(\varone_1,\ldots,\varone_n)\\
     \appTRS(\constr{\vartwo}{\lambdatwo}(\vartwo_1,\ldots,\vartwo_{m}),\vartwo)&\rewrTRSW&\LambdatoTRSW{\lambdatwo}
    \end{eqnarray*}
    where $\lambdaone$ ranges over $\lambda$-terms, $\lambdatwo$ ranges over abstractions and
    applications,
    $\FV{\lambda\varone.\lambdaone}=\varone_1,\ldots,\varone_n$
    and $\FV{\lambda\vartwo.\lambdatwo}=\vartwo_1,\ldots,\vartwo_m$. These rewrite rules
    are said to be \emph{ordinary rules}. We also need the following \emph{administrative} rule:
    $$
    \appTRS(\cappTRSW(\varone,\vartwo),\varthree)\rewrTRSW\appTRS(\appTRS(\varone,\vartwo),\varthree)
    $$
\end{varitemize}
The CTRS $\TRSW$ is slightly more complicated than $\TRS$: some additional
overhead is needed to force reduction to happen only in head position.
As usual, to every term $\termone\in\TRSWvarterms$ we can associate a term $\TRSWtolambda{\termone}$:
\begin{eqnarray*}
  \TRSWtolambda{\varone}&=&\varone\\
  \TRSWtolambda{\appTRS(\termtwo,\termthree)}=\TRSWtolambda{\cappTRSW(\termtwo,\termthree)}
     &=&\TRSWtolambda{\termtwo}\TRSWtolambda{\termthree}\\
  \TRSWtolambda{\constr{\varone}{\lambdaone}(\termone_1,\ldots\termone_n)}&=&
  (\lambda\varone.\lambdaone)\{\TRSWtolambda{\termone_1}/\varone_1,\ldots,\TRSWtolambda{\termone_n}/\varone_n\}
\end{eqnarray*}
where $\FV{\lambda\varone.\lambdaone}=\varone_1,\ldots,\varone_n$.
A term $\termone\in\TRSWterms$ is \emph{canonical} if either $\termone=\constr{\varone}{\lambdaone}(\termone_1\,\ldots,\termone_n)\in\TRSWconterms$ or
$\termone=\appTRS(\termtwo,\termthree)$ where $\termtwo$ is canonical and $\termthree\in\TRSWconterms$.
\condinc{
\begin{lemma}\label{lemma:closedcanonical}
For every closed $\lambdaone\in\Lambdaterms$, $\LambdatoTRSW{\lambdaone}$ is canonical.
\end{lemma}
\begin{proof}
By a straightforward induction on $\lambdaone$.
\end{proof}
The obvious variation on Equation~\ref{equat:commute} holds here:
\begin{equation}\label{equat:commutew}
\TRSWtolambda{\LambdatoTRSW{\lambdaone}\{\termone_1/\varone_1,\ldots,\termone_n/\varone_n\}}=
\lambdaone\{\TRSWtolambda{\termone_1}/\varone_1,\ldots,\TRSWtolambda{\termone_n}/\varone_n\}.
\end{equation}}{}
$\TRSW$ mimics call-by-name reduction in much the same way $\TRS$ mimics call-by-value
reduction. However, one reduction step in the $\lambda$-calculus corresponds to $n\geq 1$ steps
in $\TRSW$, although $n$ is kept under control:
\begin{lemma}\label{lemma:simul}
Suppose that $\termone\in\TRSWterms$ is canonical and that $\termone\rewrTRSW\termtwo$. Then
there is a natural number $n$ such that:
\begin{varenumerate}
\item
  $\TRSWtolambda{\termone}\rewrlambdah\TRSWtolambda{\termtwo}$;
\item
  There is a canonical term $\termthree\in\TRSWterms$ such that $\termtwo\rewrTRSW^n\termthree$;
\item
  $\plength{\termfour}{\appTRS}=\plength{\termtwo}{\appTRS}+m$ whenever $\termtwo\rewrTRSW^m\termfour$
  and $m\leq n$;
\item
  $\TRSWtolambda{\termfour}=\TRSWtolambda{\termtwo}$ whenever $\termtwo\rewrTRSW^m\termfour$
  and $m\leq n$.
\end{varenumerate}
\end{lemma}
\condinc{
\begin{proof}
A term $\termone$ is said to be \emph{semi-canonical} iff $\termone=\appTRS(\termtwo,\termthree)$, where
$\termthree\in\TRSWconterms$ and $\termtwo$ is either semi-canonical or is itself an element
of $\TRSWconterms$. We now prove that if $\termone$ is semi-canonical, there there are a natural
number $n$ and a canonical term $\termtwo$ such that:
\begin{varitemize}
\item
  $\termone\rewrTRSW^n\termtwo$;
\item
  $\plength{\termthree}{\appTRS}=\plength{\termone}{\appTRS}+m$ whenever $\termone\rewrTRSW^m\termthree$
  and $m\leq n$;
\item
  $\TRSWtolambda{\termthree}=\TRSWtolambda{\termone}$ whenever $\termone\rewrTRSW^m\termthree$
  and $m\leq n$.
\end{varitemize}
We can proceed by induction on $\length{\termone}$. By definition $\termone$ is always
in the form $\appTRS(\termfour,\termfive)$. We distinguish three cases:
\begin{varitemize}
\item
  $\termfour$ is semi-canonical. Then, we get what we want by induction hypothesis.
\item
  $\termfour$ is in $\TRSWconterms$ and has the form $\constr{\varone}{\lambdaone}(\termone_1,\ldots,\termone_m)$.
  Then, $n=0$ and $\termone$ is itself canonical.
\item
  $\termfour$ is in $\TRSWconterms$ and has the form $\cappTRSW(\termsix,\termseven)$. Then
  $$
  \termone=\appTRS(\cappTRSW(\termsix,\termseven),\termfive)\rewrTRSW\appTRS(\appTRS(\termsix,\termseven),\termfive)
  $$
  We can apply the induction hypothesis to $\appTRS(\termsix,\termseven)$ (since its length is
  strictly smaller than $\length{\termone}$).
\end{varitemize}
We can now proceed as in Lemma~\ref{lemma:TRStolam}, since whenever
$\termone$ rewrites to $\termtwo$ by one of the ordinary rules, $\termtwo$ is 
semi-canonical.
\end{proof}
\begin{lemma}\label{lemma:TRSWnormalform}
A canonical term $\termone\in\TRSWterms$ is in normal form iff $\TRSWtolambda{\termone}$ is in
normal form.
\end{lemma}
\begin{proof}
We first prove that any canonical normal form $\termone$ can be written
as $\constr{\varone}{\lambdaone}(\termone_1,\ldots,\termone_n)$, where
$\termone_1,\ldots,\termone_n\in\TRSWconterms$. We proceed by induction
on $\termone$:
\begin{varitemize}
\item
  If $\termone=\constr{\varone}{\lambdaone}(\termone_1,\ldots,\termone_n)$,
  then the thesis holds.
\item
  If $\termone=\appTRS(\termtwo,\termthree)$, then $\termtwo$ is canonical
  and in normal form, hence in the form $\constr{\varone}{\lambdaone}(\termone_1,\ldots,\termone_n)$
  by induction hypothesis. As a consequence, $\termone$ is not a normal
  form, which is a contraddiction.
\end{varitemize}
We can now prove the statement of the lemma, by distinguishing two cases:
\begin{varitemize}
\item
  If $\termone=\constr{\varone}{\lambdaone}(\termone_1,\ldots,\termone_n)$,
  where $\termone_1,\ldots,\termone_n\in\TRSWconterms$, then
  $\termone$ is in normal form and $\TRSWtolambda{\termone}$ is an
  abstraction, hence a normal form.
\item
  If $\termone=\appTRS(\termtwo,\termthree)$, then $\termone$ cannot be a normal form,
  since $\termtwo$ is canonical and in normal form and, as a consequence,
  it can be written as $\constr{\varone}{\lambdaone}(\termone_1,\ldots,\termone_n)$.
\end{varitemize}
This concludes the proof.
\end{proof}
Observe that this property holds only if $\termone$ is canonical: a non-canonical term
can reduce to another one (canonical or not) even if the underlying $\lambda$-term
is a normal form.
\begin{lemma}\label{lemma:lamtoTRSW}
If $\lambdaone\rewrlambdah\lambdatwo$, $\termone$ is canonical and $\TRSWtolambda{\termone}=\lambdaone$,
then $\termone\rewrTRSW\termtwo$, where $\TRSWtolambda{\termtwo}=\lambdatwo$ and 
$\plength{\termtwo}{\appTRS}+1\geq\plength{\termone}{\appTRS}$.
\end{lemma}
\begin{proof}
Similar to the one of Lemma~\ref{lemma:lamtoTRSW}.
\end{proof}
}{}
The slight mismatch between call-by-name reduction in $\Lambdaterms$ and
reduction in $\TRSW$ is anyway harmless globally: the total number of reduction
step in $\TRSW$ is at most two times as large as the total number of call-by-name reduction
steps in $\Lambdaterms$.
\begin{theorem}[Term Reducibility]
Let $\lambdaone\in\Lambdaterms$ be a closed term. The following
two conditions are equivalent:
\begin{varenumerate}
\item
  $\lambdaone\rewrlambdah^n\lambdatwo$ where $\lambdatwo$ is in normal form;
\item
  $\LambdatoTRSW{\lambdaone}\rewrTRSW^m\termone$ where
  $\TRSonetolambda{\termone}=\lambdatwo$ and $\termone$ is in normal form.
\end{varenumerate}
Moreover $n\leq m\leq 2n$.
\end{theorem}
\condinc{
\begin{proof}
Suppose $\lambdaone\rewrlambdah^n\lambdatwo$, where $\lambdatwo$ is in normal form.
$\lambdaone$ is closed and, by Lemma~\ref{lemma:closedcanonical}, $\LambdatoTRSW{\lambdaone}$
is canonical. By iterating over Lemma~\ref{lemma:simul} and Lemma~\ref{lemma:lamtoTRSW},
we obtain the existence of a term $\termone$ such that $\TRSWtolambda{\termone}=\termtwo$,
$\termone$ is in normal form and $\LambdatoTRSW{\lambdaone}\rewrTRSW^m\termone$, where $m\geq n$ and
$$
\plength{\termone}{\appTRS}-\plength{\LambdatoTRSW{\lambdaone}}{\appTRS}\geq (m-n)-n.
$$
Since $\plength{\termone}{\appTRS}=0$ ($\termone$ is in normal form), $m\leq 2n$.
If $\LambdatoTRSW{\lambdaone}\rewrTRSW^m\termone$ where
$\TRSonetolambda{\termone}=\lambdatwo$ and $\termone$ is in normal form, then
by iterating over Lemma~\ref{lemma:simul} we obtain that $\lambdaone\rewrlambdah^n\lambdatwo$
where $n\leq m\leq 2n$ and $\lambdatwo$ is in normal form.
\end{proof}}{}
$\GRSW$ is the graph rewrite system corresponding to $\TRSW$, in the sense of
\condinc{Section~\ref{Sect:GraphRep}}{Section~\ref{Sect:Byproduct}}. Exactly as for the call-by-value case, computing the normal
form of (the graph representation of) any term takes time polynomial in the
number of reduction steps to normal form:
\begin{theorem}\label{thm:mainw}
There is a polynomial $p:\N^2\rightarrow\N$ such that for every $\lambda$-term $\lambdaone$,
the normal form of $\TRStoGRSW{\LambdatoTRSW{\lambdaone}}$ can be computed in time at most 
$p(|\lambdaone|,\Timew{\lambdaone})$.
\end{theorem}
On the other hand, we cannot hope to \emph{directly} reuse the results in Section~\ref{Sect:CTR2L}
when proving the existence of an embedding of CRSs into weak call-by-name 
$\lambda$-calculus: the same $\lambda$-term can have distinct normal forms in the two cases. 
It is widely known, however, that a continuation-passing translation
can be used to simulate call-by-value reduction by call-by-name reduction~\cite{Plotkin75tcs}.
The only missing tale is about the relative performances: do terms obtained via the CPS 
translation reduce (in call-by-name) to their normal forms in a number of
steps which is \emph{comparable} to the number of (call-by-value) steps to normal form
for the original terms? We conjecture the answer is ``yes'', but we leave the task of proving that
to a future work.}{}
\section{Conclusions}
We have shown that the most na\"\i{}ve cost models for weak call-by-value 
and call-by-name
$\lambda$-calculus
(each beta-reduction step has unitary cost)
and orthogonal constructor term rewriting (each rule application has unitary cost)
are linearly related. Since, in turn, this cost model for
$\lambda$-calculus is polynomially related to the actual cost of reducing
a $\lambda$-term on a Turing machine, 
the two machine models we considered are both \emph{reasonable} machines, when
endowed with their natural, intrinsic cost models (see also Gurevich's opus on Abstract State Machine
simulation ``at the same level of abstraction'', e.g.~\cite{Gurevich}).
This strong (the embeddings we consider are compositional), complexity-preserving
equivalence between a first-order and a higher-order model is the most important technical
result of the paper. 

Ongoing and future work includes the investigation of how much of this simulation could be
recovered either in a typed setting (see~\cite{SplawskiU99} for some of the difficulties),
or in the case of $\lambda$-calculus with strong reduction, where we reduce under an abstraction. 
Novel techniques have to be developed, since the analysis we performed in the present 
paper cannot be easily extended to these cases.

\subsection*{Acknowledgments}
The authors wish to thank Kazushige Terui for stimulating
discussions on the topics of this paper.

\condinc{\bibliographystyle{plain}}{\bibliographystyle{plain}}
\bibliography{wie}
\end{document}

%% file: macros.tex
\usepackage{proof}
\usepackage{amsmath}
\usepackage{subfigure}
\usepackage{epsf}
\usepackage{graphics}
\usepackage{wrapfig}
\usepackage{mathrsfs}
\usepackage{url}
\condinc{\usepackage{a4wide}}{}
\usepackage{calc}
\usepackage{amsmath,amssymb,amsfonts,mathrsfs,latexsym,stmaryrd}
\usepackage[all]{xy}

\newcommand{\funone}{\mathbf{f}}

\newcommand{\conone}{\mathbf{c}}

\newcommand{\patone}{\mathbf{p}}
\newcommand{\pattwo}{\mathbf{r}}
\newcommand{\seqpone}{\alpha}
\newcommand{\seqptwo}{\beta}
\newcommand{\varone}{x}
\newcommand{\vartwo}{y}
\newcommand{\varthree}{z}
\newcommand{\varfour}{w}
\newcommand{\varfive}{f}
\newcommand{\varsix}{g}
\newcommand{\varseven}{h}
\newcommand{\lambdaone}{M}
\newcommand{\lambdatwo}{N}
\newcommand{\lambdathree}{L}
\newcommand{\lambdafour}{P}
\newcommand{\lambdafive}{Q}
\newcommand{\valueone}{V}
\newcommand{\valuetwo}{W}

\newcommand{\termone}{t}
\newcommand{\termtwo}{u}
\newcommand{\termthree}{v}
\newcommand{\termfour}{w}
\newcommand{\termfive}{d}
\newcommand{\termsix}{e}
\newcommand{\termseven}{f}
\newcommand{\cltermone}{X}
\newcommand{\cltermtwo}{Y}
\newcommand{\sigone}{\Sigma}

\newcommand{\vsone}{V}

\newcommand{\ordone}{\alpha}

\newcommand{\labelone}{\delta}

\newcommand{\rootone}{r}
\newcommand{\roottwo}{s}
\newcommand{\verone}{v}

\newcommand{\tgone}{G}
\newcommand{\tgtwo}{H}
\newcommand{\tgthree}{J}
\newcommand{\tgfour}{K}
\newcommand{\tgfive}{I}
\newcommand{\rrone}{\rho}

\newcommand{\homone}{\varphi}

\newcommand{\nat}[1]{\lceil #1\rceil}
\newcommand{\rewrite}[1]{\stackrel{#1}{\longrightarrow}}

\newcommand{\appTRS}{\mathbf{app}}
\newcommand{\cappTRSW}{\mathbf{capp}}
\newcommand{\constr}[2]{\mathbf{c}_{#1,#2}}
\newcommand{\LambdatoTRS}[1]{[#1]_{\Phi}}
\newcommand{\TRStolambda}[1]{\langle #1\rangle_{\Lambdaterms}}
\newcommand{\TRSonetolambda}[1]{\langle\!\langle #1\rangle\!\rangle_{\Lambdaterms}}
\newcommand{\TRSonetolambdaII}[1]{[#1]_{\Lambdaterms}}
\newcommand{\TRSonetolambdaIII}[1]{\langle\!|#1|\!\rangle_{\Lambdaterms}}
\newcommand{\TRSWtolambda}[1]{\langle #1\rangle_{\Lambdaterms}}
\newcommand{\TRStoGRS}[1]{[#1]_{\Theta}}
\newcommand{\TRStoGRSW}[1]{[#1]_{\Xi}}
\newcommand{\GRStoTRS}[1]{\langle #1\rangle_{\Phi}}
\newcommand{\LambdatoTRSW}[1]{[#1]_{\Psi}}
\newcommand{\LambdatoTRSWaux}[1]{\{#1\}_{\Psi}}
\newcommand{\errorterm}{\bot}
\newcommand{\arity}[1]{\mathit{ar}(#1)}
\newcommand{\dsize}[1]{||#1||}

\newcommand{\Variables}{\Upsilon}
\newcommand{\Lambdaterms}{\Lambda}

\newcommand{\Functions}[1]{\Sigma_{#1}}
\newcommand{\Rules}[1]{\mathcal{R}_{#1}}
\newcommand{\Graphs}[1]{\mathcal{G}_{#1}}
\newcommand{\TRS}{\Phi}
\newcommand{\TRSterms}{\mathcal{T}(\Phi)}
\newcommand{\TRSvarterms}{\mathcal{V}(\Phi,\Variables)}
\newcommand{\TRSconterms}{\mathcal{C}(\Phi)}
\newcommand{\TRSone}{\Xi}

\newcommand{\TRStermsp}[1]{\mathcal{T}(#1)}
\newcommand{\TRSpatsp}[1]{\mathcal{P}(#1,\Variables)}
\newcommand{\TRSvartermsp}[1]{\mathcal{V}(#1,\Variables)}
\newcommand{\TRScontermsp}[1]{\mathcal{C}(#1)}
\newcommand{\GRS}{\Theta}
\newcommand{\GRSW}{\Xi}
\newcommand{\CtoCG}[1]{[#1]_{\Graphs{}}}
\newcommand{\CGtoC}[1]{\langle #1\rangle_{\Rules{}}}
\newcommand{\TRSW}{\Psi}
\newcommand{\TRSWterms}{\mathcal{T}(\Psi)}
\newcommand{\TRSWvarterms}{\mathcal{V}(\Psi,\Variables)}
\newcommand{\TRSWconterms}{\mathcal{C}(\Psi)}
\newcommand{\domain}[1]{\mathit{dom}(#1)}
\newcommand{\srrone}{\Rules{}}

\newcommand{\sgrone}{\Graphs{}}
\newcommand{\Time}[1]{\mathit{Time}_v(#1)}
\newcommand{\Timew}[1]{\mathit{Time}_h(#1)}

\newcommand{\TMone}{\mathcal{M}}

\newcommand{\rewrTRS}{\rightarrow}
\newcommand{\rewrTRSW}{\rightarrow}
\newcommand{\rewr}[1]{\rightarrow_{#1}}
\newcommand{\rewrlambdav}{\rightarrow_v}
\newcommand{\rewrlambdah}{\rightarrow_h}
\newcommand{\rewrgraph}{\rightarrow}

\newcommand{\subgr}[2]{#1\downarrow #2}

\newcommand{\cod}[1]{\ulcorner #1\urcorner}

\condinc{\newcommand{\N}{\mathbb{N}}}{\newcommand{\N}{\mathbb{N}}}

\newcommand{\FV}[1]{\mathtt{FV}(#1)}
\newcommand{\length}[1]{|#1|}
\newcommand{\plength}[2]{|#1|_{#2}}

\condinc{
\newtheorem{lemma}{Lemma}
\newtheorem{proposition}{Proposition}
\newtheorem{theorem}{Theorem}
\newtheorem{corollary}{Corollary}
\newenvironment{proof}{\begin{trivlist}
       \item[\hskip \labelsep {\bfseries Proof.}]}{\hfill $\Box$ \end{trivlist}}
\newtheorem{definition}{Definition}
\newtheorem{example}{Example}
}{}

\condinc{
\newenvironment{varitemize}
{
\begin{list}{\labelitemi}
{\setlength{\itemsep}{0.0mm}
 \setlength{\topsep}{0.0mm}
 \setlength{\parindent}{0.0mm}
 \setlength{\parskip}{0.0mm}
 \setlength{\parsep}{0.0mm}
 \setlength{\partopsep}{0.0mm}
 \setlength{\leftmargin}{15pt}
 \setlength{\labelsep}{5pt}
 \setlength{\labelwidth}{10pt}}}
{
 \end{list} 
}}{
\newenvironment{varitemize}
{
\begin{list}{\labelitemii}
{\setlength{\itemsep}{0.0mm}
 \setlength{\topsep}{0.0mm}
 \setlength{\parindent}{0.0mm}
 \setlength{\parskip}{0.0mm}
 \setlength{\parsep}{0.0mm}
 \setlength{\partopsep}{0.0mm}
 \setlength{\leftmargin}{15pt}
 \setlength{\labelsep}{5pt}
 \setlength{\labelwidth}{10pt}}}
{
 \end{list} 
}}

\newcounter{number}

\newenvironment{varenumerate}
{\begin{list}{\arabic{number}.}
  {
   \usecounter{number}
   \setlength{\labelwidth}{4.0mm}
   \setlength{\labelsep}{2.0mm}
   \setlength{\itemindent}{0.0mm}
   \setlength{\itemsep}{0.0mm}
   \setlength{\topsep}{0.0mm}
   \setlength{\parskip}{0.0mm}
   \setlength{\parsep}{0.0mm}
   \setlength{\partopsep}{0.0mm}
  }
}
{\end{list}}